\documentclass[twoside,11pt]{article}

\usepackage{jmlr2e}
\usepackage{amsmath}
\usepackage{amssymb}
\usepackage{amsfonts}
\usepackage{graphicx}
\usepackage{tikz} 
\usepackage{geometry}
\usepackage{algorithm}
\usepackage{algpseudocode}

\usepackage{lastpage}
\jmlrheading{23}{2022}{1-\pageref{LastPage}}{10/21; Revised
4/22}{6/22}{21-1262}{Michael Pearce and Elena A. Erosheva}

\ShortHeadings{A Unified Model for Rankings and Scores}{Pearce and Erosheva}
\firstpageno{1}

\begin{document} 

\title{A Unified Statistical Learning Model for Rankings and Scores with Application to Grant Panel Review}

\author{\name Michael Pearce \email mpp790@uw.edu \\
       \addr Department of Statistics\\
       University of Washington\\
       Seattle, WA 98195-4322, USA
       \AND
       \name Elena A. Erosheva \email erosheva@uw.edu \\
       \addr Department of Statistics, School of Social Work, and the Center for Statistics and the Social Sciences\\
       University of Washington\\
       Seattle, WA 98195-4322, USA}

\editor{Gal Elidan}

\maketitle

\begin{abstract}
Rankings and scores are two common data types used by judges to express preferences and/or perceptions of quality in a collection of objects. Numerous models exist to study data of each type separately, but no unified statistical model captures both data types simultaneously without first performing data conversion. We propose the Mallows-Binomial model to close this gap, which combines a Mallows $\phi$ ranking model with Binomial score models through shared parameters that quantify object quality, a consensus ranking, and the level of consensus among judges. We propose an efficient tree-search algorithm to calculate the exact MLE of model parameters, study statistical properties of the model both analytically and through simulation, and apply our model to real data from an instance of grant panel review that collected both scores and partial rankings. Furthermore, we demonstrate how model outputs can be used to rank objects with confidence. The proposed model is shown to sensibly combine information from both scores and rankings to quantify object quality and measure consensus with appropriate levels of statistical uncertainty.
\end{abstract}


\section{Introduction}\label{sect:introduction}
Preference data is common to our world: Citizens express preferences through voting in elections, critics rank movies when creating annual top-10 lists, judges score figure skaters in the Olympics using numerical scales, wine critics use Likert scales with words such as ``mediocre'' to rate wines, consumers use stars to convey the quality of a product, and so on. 
As can be seen in these examples, preferences appear in different forms: most commonly as \textit{rankings} or \textit{scores}. Rankings denote a relative order of objects from best to worst, potentially allowing ties; ranks refer to the specific place of each object in the ranking. Scores are numerical values given to objects to denote their quality. Scores provide more granular information than rankings through the relative distance between scores and the rankings they induce. A number of existing models \citep{Mallows1957,Fligner1986,Rost1988} have been studied to model rankings and scores individually. In these models, a natural goal is to aggregate preferences into a \textit{consensus ranking}, which expresses the overall preferences of those providing the rankings, whom we call judges. Statistical models may additionally incorporate information regarding the uncertainty of those rankings, or the level of global or local consensus among the judges.

When both rankings and scores are available, incorporating the additional information may improve accuracy of preference aggregation. This is partly because rankings and scores provide distinct types of information. Scores allow judges to express their overall perceptions on a numerical scale and make relative comparisons implicitly. However, some judges may be naturally more lenient or harsh (e.g., one judge's 5/10 may be another's 7/10), or may become cognitively burdened by the number of scores they need to provide and stop expressing internally consistent scores \citep{johnson2008statistical,Wang2018}. For these reasons, scores have sometimes been described as highly subjective and inconsistent in different bodies of literature \citep{biernat1995shifting,biernat1997gender,biernat2009race,mallard2009fairness}. On the other hand, rankings are thought to provide objective comparisons between ranked objects because they require the judge to make explicit comparisons that are scale-free \citep{biernat1995shifting}. For example, we can be more confident that a judge truly perceives object A above object B if A is ranked above B, which may not be true for scores. However, rankings force demarcation even when it does not exist and lack granularity in comparisons (e.g., first and second place are nearly tied, yet third place is far behind). As can be seen, scores and rankings have complementary pros and cons. Therefore, regardless of the amount of each data type present (or missing), incorporating both types into a single model has been suggested as a potential “best of both worlds” approach to solve an existing false dichotomy by incorporating the benefits of each preference data type \citep{shah2018design,ovadia2004ratings}.

Rankings and scores arise simultaneously in a number of scenarios. For example, in peer review, judges may score proposals numerically and subsequently rank their top few favorites. Another example arises during information retrieval when relevancy criteria from different sources and types---such as from an algorithmic database search for relevant documents or from human judgment---are available \citep{Hsu2005}. In the former example, judges provide both kinds of information, while in the latter, one system may provide rankings and another may provide scores. Either way, it is unclear how the information from each preference data type should be utilized. \cite{Renda2003} suggest that neither ranking or score data is uniformly better than the other when analyzed alone (in the context of metasearch), and many authors have argued that using both is better \citep{Belkin1995,Macdonald2009,Lee1997,Balog2012}. To incorporate both sources of information, conversion of data from one type to another is common \citep{Hsu2005,Bhamidipati2008,Li2009,shah2018design}. These methods will be discussed in Section \ref{sect:relatedwork}. Yet, to the best of our knowledge, there are no statistical models that combine information from both rankings and scores without data conversion. Although convenient, data conversion may result in a loss of information or affect results depending on the chosen conversion process, which is suboptimal.

In this paper, we propose a unified model to capture information from both rankings and scores when applied to a finite collection of objects. Conditional on the objects' true underlying qualities, Binomial scores and Mallows rankings have independent error distributions that reflect the distinct tasks of formulating scores and rankings and allow for judges to be internally inconsistent when expressing different types of preferences. Model parameters can be used to quantify both the absolute and relative qualities of the objects, identify a consensus ranking, and measure the strength of consensus using an existing metric in the literature. We formulate exact and approximate algorithms to find maximum likelihood estimators of the model parameters and demonstrate regimes in which each may be useful. In addition to simulation studies, we apply the model to real data from grant panel review in which scores and rankings were collected from the same judges. We show how the estimated parameters can be used to learn the rank ordering of grant proposals and the associated statistical uncertainty to make funding decisions.

The rest of this paper is organized as follows. After discussing related work in Section \ref{sect:relatedwork}, we state our proposed model and describe its statistical properties in Section \ref{sect:method}. We propose exact and approximate frequentist estimation algorithms in Section \ref{sect:estimation}, and compare their speed and accuracy on simulated data in Section \ref{sect:simulation}. We illustrate the model application on real ranking and score data collected during the Fall 2020 grant panel review cycle at the American Institute of Biological Sciences (AIBS) in Section 
\ref{sect:realdata}. We conclude with a discussion in Section \ref{sect:discussion}.

\section{Related Work}\label{sect:relatedwork}

A variety of methods can be used to aggregate ranking data alone, such as the Bradley-Terry, Plackett-Luce, and Mallows models \citep{Bradley1952,Plackett1975,Mallows1957}. The latter model and its extensions have been particularly well-studied in recent decades. In a seminal work, \cite{Fligner1986} state statistical qualities of the Mallows $\theta$ and $\phi$ models, which are based on the correlation coefficients of \cite{Spearman1904} and \cite{Kendall1938}, respectively. The Mallows $\phi$ model has received particular attention as a natural fit in many ranking applications. Henceforth referred to simply as the Mallows model, it is a location-scale probability distribution that measures distance between rankings using Kendall's $\tau$, which can be defined as the minimum number of adjacent swaps between objects needed to convert one ranking into another, or, equivalently, as the number of discordant object pairs between two rankings (as defined in Equation \ref{Eq:Kendall}). Many extensions exist, such as the Generalized Mallows model which permits unique scale parameters at each ranking level \citep{Fligner1988} and the Infinite Generalized Mallows model which aggregates rankings over infinite collections of objects \citep{Meila2010}. 

For rankings, we focus on the original Mallows model and its partial-ranking counterpart (also proposed in \citeauthor{Fligner1986}, \citeyear{Fligner1986}). These distributions can model the following situation: Suppose a judge scores a collection of $J$ objects and provides a top-$R$ ranking, $R\leq J$. The ranking is called \textit{partial} when $R<J$ and \textit{complete} when $R=J$. We denote his/her ranking by $\Pi=(\pi(1),\dots,\pi(R))$, where $r\in\{1,\dots,R\}$ is the \textit{rank} of object $\pi(r)$ (i.e., $\pi(r)$ is the $r^\text{th}$ most preferred object). If an object is not ranked, it is assumed to be perceived worse than any ranked object; any pair of unranked objects cannot be compared. Then, for fixed $R$, the probability of observing a specific ranking $\pi$ of length $R$ under the Mallows model with consensus ranking $\pi_0$ of length $J$ and scale parameter $\theta$ is
\begin{align}
    P[\Pi=\pi | \pi_0,\theta] &= \frac{e^{-\theta d_{R,J}(\pi,\pi_0)}}{\psi_{R,J}(\theta)}\label{Eq:Mallows}\\
    d_{R,J}(\pi,\pi_0) &= \sum_{1\leq j_1<j_2\leq J} I\{j_1\prec_{\pi} j_2 \text{ and } j_2\prec_{\pi_0} j_1\} \label{Eq:Kendall}\\
    \psi_{R,J}(\theta) &= \prod_{j=1}^R \frac{1-e^{-\theta(J-j+1)}}{1-e^{-\theta}},\label{Eq:psi}
\end{align}
where $d_{R,J}(\cdot,\cdot)$ is  the Kendall distance between rankings of length $R$ and $J$, respectively, assembled over the same set of $J$ objects, and the notation $j_1\prec_{\pi}j_2$ means that object $j_1$ is ranked strictly better in $\pi$ than object $j_2$ (pairs of unranked objects do not satisfy this operator). The model is exponential; the consensus ranking $\pi_0$ is the modal probability ranking. When $\theta$ is large, the distribution of rankings is concentrated around $\pi_0$, and as $\theta$ approaches 0, the distribution approaches a uniform over the permutations of length $R$.

On the other hand, it is rare for scores to be modeled statistically. Instead, simple summary statistics such as the mean or median are commonly used \citep{lee2013bias,tay2020beyond,NIHPeerReview}. Scores often arise from a discrete, ordinal, and equally-spaced set with minimum and maximum allowable values. These qualities make a number of common univariate probability distributions inappropriate (e.g., continuous distributions, Poisson, geometric, etc.). In this case, the allowable set of scores can be linearly transformed into the set of integers $\{0,1,\dots, M\}$. The Binomial distribution is therefore a theoretically reasonable and convenient choice as a well-known and unimodal distribution with a natural mean parameter. We model the score, $X$, for any given object using the Binomial($M,p$) distribution,
\begin{align}
    P[X=x|M,p] &= {M \choose x}p^x(1-p)^{M-x},\label{Eq:binomialdensity}
\end{align}
where the binomial probability $p$ is called the object's \textit{true underlying quality} and $Mp$ its \textit{expected score}.

The aforementioned distributions cannot model rankings and scores together. In the social and health sciences, the literature on \textit{mixed-outcomes} includes proposed methods for combining preference data of different types via \textit{conversion}, such as converting rankings into scores or scores into rankings prior to performing a statistical analysis \citep{Thurstone1927,Salomon2003,Kim2015,Venkatraghavan2019}. The field of information retrieval within computer science also includes a growing literature on utilizing both ranking and score data in the context of \textit{data fusion} \citep{Fagin2002, Hsu2005, Bhamidipati2008, Li2009}. Although some authors have argued that many data fusion methods can be theoretically justified as probabilistic since they often estimate likelihoods of relevance during information retrieval \citep{Belkin1995}, most models are not explicitly considered as such. Across fields, a number of authors suggest that using both rankings and scores when available is generally better at eliciting accurate preference aggregation than using any single data type individually \citep{Belkin1995, Lee1997, Renda2003, ovadia2004ratings, Macdonald2009, Balog2012, shah2018design}. 

The existing literature lacks a unified statistical model for ranking and score preference data. In the next section, we propose the first such model.

\section{A Statistical Model for Rankings and Scores}\label{sect:method}

We now propose the Mallows-Binomial model and then discuss its statistical properties.

\subsection{The Mallows-Binomial Model}
Suppose a judge assesses $J$ objects using both rankings and scores. We assume that each object $j\in\{1,\dots,J\}$ has a true underlying \textit{quality}, $p_j\in[0,1]$. We use the convention that lower values of $p_j$ denote better quality. Let $X=[X_1 \ X_2 \ \dots X_J]^T$ be a vector of integer scores, where each $X_{j}\in\{0,1,\dots,M\}$ is the score assigned to object $j$. Let $\Pi$ be the top-$R$ ranking of the objects, $R\leq J$, such that no ties are allowed. $\Pi$ is called a \textit{partial ranking} when $R<J$ and a \textit{complete ranking} when $R=J$. Rankings need not align with the order of the scores.

We propose a joint probability model for the judge's ranking $\Pi$ and scores $X$,
\begin{align}
    P\big[\Pi=\pi,X=x|p,\theta\big] =& \frac{e^{-\theta d_{R,J}(\pi,\pi_0)}}{\psi_{R,J}(\theta)}\times \prod_{j=1}^J {M\choose x_{j}}p_j^{x_{j}}(1-p_j)^{M-x_{j}} \label{Eq:model}\\
    \ p=[p_1 \ \dots \ p_J]^T\in[0,&1]^J,\ \pi_0 = \text{Order}(p),\theta>0,\nonumber\\
    X_{1},\dots,X_{J},\Pi \text{ are all}&\text{ mutually independent,}\nonumber
\end{align}
where $d_{R,J}(\cdot,\cdot)$ is the Kendall's $\tau$ distance between two rankings and $\psi_{R,J}(\theta)$ is the normalizing constant of a (partial) Mallows model, as seen in Equation \ref{Eq:psi}. We refer to this model as the \textit{Mallows-Binomial}($p,\theta$) distribution.

A key aspect of this model is the incorporation of two distinct types of preference data. It can be seen directly from Equation \ref{Eq:model} that our model corresponds to $J+1$ joint observations per judge, with $J$ scores and one (partial) ranking. The Mallows-Binomial model incorporates information from both data types without conversion to learn object quality parameters, $p_j$, $j=1,\dots,J$. The joint likelihood ties together the scores and ranking by assuming that the modal consensus ranking of the Mallows component is the same as the ranking induced by the Binomial score parameters, $p_j$, $j=1,\dots,J$. This formulation naturally reflects the relationship between scores and rankings given each object's true underlying quality and the order of all objects induced by their true underlying qualities.  The parameter $\theta$ is the \textit{consensus scale parameter}, which can be interpreted exactly as in the Mallows model with respect to the rankings: Large values of $\theta$ suggest strong ranking consensus among judges. As $\theta$ decreases to 0, the model approaches a uniform distribution over the possible rankings. The Mallows-Binomial model constitutes a proper probability distribution as the product of $J+1$ independent component distributions given the parameters $(p,\theta)$.

The Mallows-Binomial model does not assume that each ranking is of the same length, that the scores and ranking of each judge align, or even that the same judges provide both rankings and scores. The only assumption is that both rankings and scores reflect the objects' true underlying qualities. 
Inconsistent preferences arise in the peer review context considered in Section \ref{sect:realdata}. In our grant peer review data, judges first score objects (grant proposals) and openly share their scores during a panel discussion, and then provide a separate partial ranking after the discussion of all objects is completed. The partial ranking is made in private, potentially leading to changes in perception of quality. Inconsistent preferences may also arise when scores and rankings are provided by different sets of judges. For example, in database search or information retrieval, relevancy criteria used by algorithms may arise from completely separate systems, such as when one system (e.g., a machine learning algorithm) provides numerical scores and another (e.g., a human judge) ranks the most relevant objects. Such situations do not affect estimation or interpretation of estimated parameters; our model can still capture distinct preferences.

\subsection{Statistical Properties}\label{sect:properties}

\subsubsection{Identifiability}

We prove that the Mallows-Binomial($p,\theta$) model is identifiable via Proposition \ref{prop:ident}.

\begin{proposition}
\label{prop:ident}
Let $M$, $J$, and $R$ be fixed and positive integers such that $R\leq J$. Then the Mallows-Binomial($p,\theta$) model is identifiable.
\end{proposition}

\begin{proof}
Let $P_{\theta,p}$ denote the probability distribution of scores $x$ and rankings $\pi$ under a Mallows-Binomial($p,\theta$) model. Let $\theta_1,\theta_2>0$ and $p_1,p_2\in[0,1]^J$ such that $P_{\theta_1,p_1}=P_{\theta_2,p_2}$. Then,
\begin{align*}
    P_{\theta_1,p_1}&=P_{\theta_2,p_2}\\
    \iff &\frac{e^{-\theta_1 d_{R,J}(\pi,\text{Order}(p_1))}}{\psi_{R,J}(\theta_1)} \prod_{j=1}^J p_{1j}^{x_{j}}(1-p_{1j})^{M-x_{j}} = \frac{e^{-\theta_2 d_{R,J}(\pi,\text{Order}(p_2))}}{\psi_{R,J}(\theta_2)} \prod_{j=1}^J p_{2j}^{x_{j}}(1-p_{2j})^{M-x_{j}}\\
    \iff& 0 = \big(\theta_2d_{R,J}(\pi,\text{Order}(p_2)) - \theta_1d_{R,J}(\pi,\text{Order}(p_1))\big) + \log\frac{\psi_{R,J}(\theta_2)}{\psi_{R,J}(\theta_1)} + \\
    & \ \ \ \ \ \ \ \sum_{j=1}^J\big[ x_j\log \frac{p_{1j}}{p_{2j}} + (M-x_j)\log\frac{1-p_{j1}}{1-p_{j2}}\big].
\end{align*}
For each $j = 1,\dots,J$, and for any arbitrary $x_j$, the expression $x_j\log \frac{p_{1j}}{p_{2j}} + (M-x_j)\log\frac{1-p_{j1}}{1-p_{j2}}=0$ if and only if $p_{1j}=p_{2j}$. Thus, for any arbitrary collection $x_1,\dots,x_J$, the final sum is 0 if and only if $p_1=p_2$. Continuing under the assumption that $p_1=p_2$, we have Order($p_1$)=Order($p_2$) and thus,
\begin{align*}
     P_{\theta_1,p_1}=P_{\theta_2,p_2} \iff& 0 = d_{R,J}(\pi,\text{Order}(p_1))(\theta_2 - \theta_1) + \log\frac{\psi_{R,J}(\theta_2)}{\psi_{R,J}(\theta_1)}
\end{align*}
which for any arbitrary $\pi$ is 0 if and only if $\theta_1=\theta_2$. Therefore, the Mallows-Binomial model is identifiable.
\end{proof}

\subsubsection{Bias and Consistency}

Bias and consistency of maximum likelihood estimators (MLE) in the Mallows and Binomial distributions is a natural starting point to examine bias and consistency of the MLE in
the combined model. \cite{Tang2019} demonstrated that in the Mallows model, the MLE $\hat\pi_0$ of the consensus ranking $\pi_0$ is consistent whereas its bias is difficult to quantify due to the categorical nature of the parameter, and $\hat\theta$ is biased upward for any number of samples but consistent as the number of judges, $I$, increases to infinity. As a univariate exponential family, $\hat p$ in a Binomial($M,p$) distribution with $M$ known is unbiased and consistent. Therefore, we expect Mallows-Binomial($p,\theta$) MLEs $\hat p$ and $\hat\theta$ to be consistent but potentially biased.

It is straightforward to prove that $\hat\theta$ is biased upward since $\hat\theta=\infty$ whenever all rankings are identical to $\hat\pi_0=\text{Order}(\hat p)$, which occurs with positive probability for any $\theta\in(0,\infty)$. However, excluding such situations, bias is difficult to demonstrate. An illustration of minimal but present bias can be found in Appendix \ref{sect:A_bias}. On the other hand, Proposition \ref{prop:consistency} gives consistency of the MLEs $(\hat p,\hat\theta)$ in the Mallows-Binomial model.

\begin{proposition}
\label{prop:consistency}
Suppose $M$, $J$, and $R$ are known and let $\theta\in(0,\infty)$ and $p\in(0,1)^J$. Let $(X,\Pi)_I$ denote a sample of $I$ independent and identically distributed samples from a Mallows-Binomial($p,\theta$) distribution, and $(\hat p,\hat\theta)_I$ be the maximum likelihood estimators based on that sample. Then, $(\hat p,\hat\theta)_I\overset{p}{\to} (p,\theta)$.
\end{proposition}
A technical proof of Proposition \ref{prop:consistency} is relegated to Appendix \ref{sect:A_consist}.

Since the magnitude of bias and the rate of convergence are challenging to derive analytically, we explore these concepts through simulation. We ran simulations for different values of the model constants: the number of judges $I\in\{5,20,80\}$, maximum integer score $M\in\{10,20,40\}$, number of objects $J\in\{6,12,18\}$, and size of ranking $R\in\{6,12,18 | R\leq J\}$. Then, for each unique combination of $I$, $M$, $J$, and $R$, we performed 20 simulations for each value of $\theta\in\{1,2,3\}$, where in  each simulation we sampled a new object quality vector $p$ from a Uniform$[0,1]^J$. After examining results separately for different values of $I,M,J$, and $R$, we noticed minimal differences based on $M$ or $R$. Therefore, we present aggregated results for given $I$ and $J$ in Figure \ref{fig:bias_cons}.

\begin{figure}[h!]
    \centering
    \includegraphics[width=\textwidth]{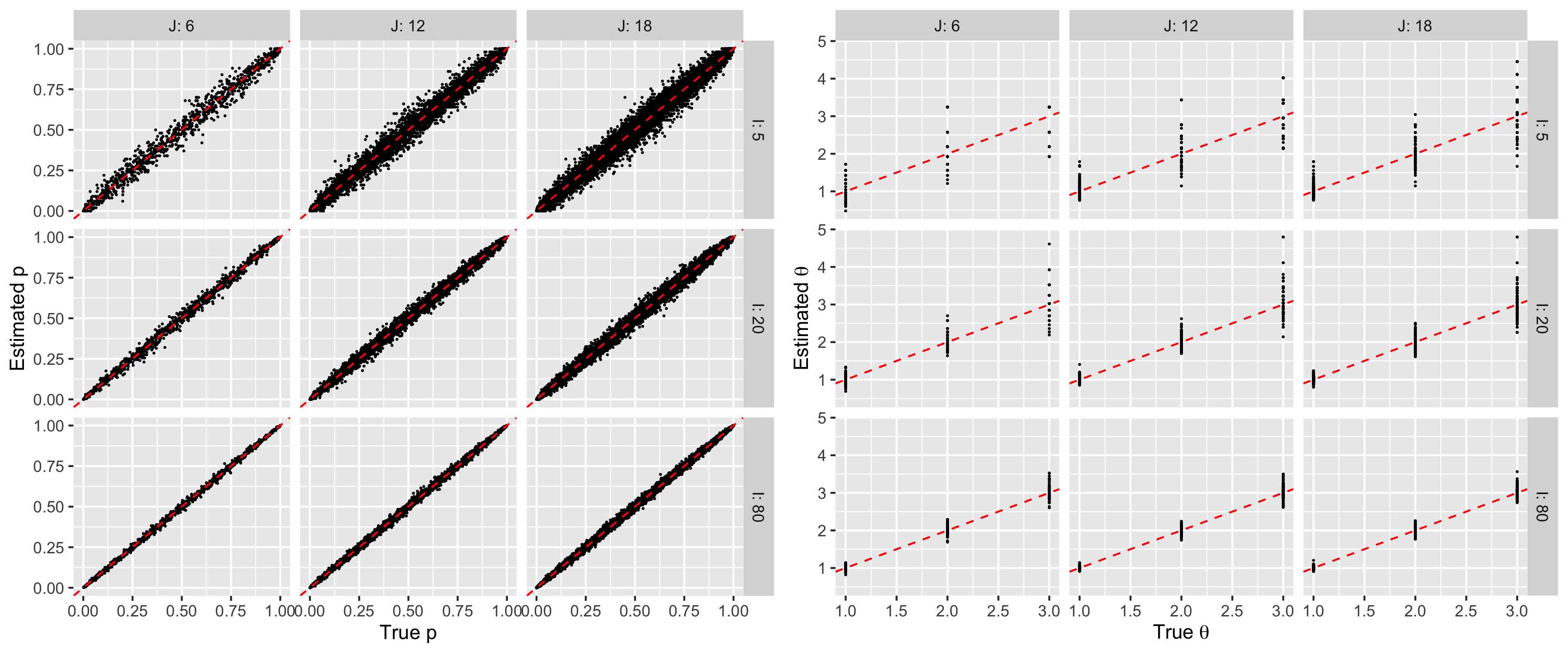}
    \caption{Simulated bias and consistency of $\hat p$ (\textit{left}) and $\hat\theta$ (\textit{right}) in Mallows-Binomial model across different values of $I$ and $J$. Red dotted lines represent perfect estimation accuracy. Results are aggregated over $M$ and $R$; the right panel excludes cases where all sampled ranking data was uniform ($\hat\theta=\infty$).}
    \label{fig:bias_cons}
\end{figure}

These simulations indicate that model parameters appear unbiased and consistent in $I$. The parameters are at worst minimally biased and exhibit estimation uncertainty similar in scale to that when estimating Binomial probabilities or Mallows scale parameters in independent models, even for modest numbers of judges, $I$.

\subsubsection{Standard Errors}

We propose estimating standard errors via the nonparametric bootstrap \citep{Efron1994}. It can be challenging to prove the asymptotic validity of bootstrap estimators in a complex setting such as this, and such results will not be provided here. However, the consistency of the MLEs $(\hat p,\hat\theta)$ in tandem with the asymptotic normality of the estimators $\hat p$ in Binomial models and $\hat\theta$ in standard Mallows models \citep{Fligner1986} suggests the bootstrap will be valid in the Mallows-Binomial model. Due to the presence of $J+1$ parameters, we recommend a relatively large number of bootstrap samples in order to obtain a proper empirical distribution of the estimators.  Due to the constrained parameter domains, we suggest constructing confidence intervals for $\hat p$ and $\hat\theta$ using a percentile-based approach rather than one based on normality assumptions.
As in the case of bias and consistency, analytic results for calculating standard errors are challenging and beyond the scope of this paper.

Estimating uncertainty in the estimated consensus ranking is paramount to many preference aggregation tasks, such as in the grant panel review application seen in Section \ref{sect:realdata}. However, bootstrapped confidence intervals for $\hat p$ and $\hat \theta$ do not directly provide confidence intervals for the estimated consensus  ranking of objects, $\hat\pi_0$. To create confidence intervals for consensus rankings, we again propose using the nonparametric bootstrap. Specifically, for each bootstrap sample and the associated MLE, the order of the estimated object quality parameters can be treated as one observation in the empirical distribution of the estimated consensus ranking. We can subsequently form confidence intervals from the empirical distribution in a straightforward manner. Conveniently, the same bootstrap samples used when creating confidence intervals for $\hat p$ and $\hat \theta$ may be used again here for computational efficiency.

\section{Estimation}\label{sect:estimation}

Analytic solutions for the maximum likelihood estimator (MLE) of a Mallows distribution do not exist. Even more, finding the MLE is an NP-hard problem \citep{Meila2012}. Difficulty arises from the discrete consensus ranking, which may be one of $J!$ unique possibilities. Although the Mallows-Binomial model contains $J+1$ \textit{continuous} parameters, $(p,\theta)\in[0,1]^J\times \mathbb{R}_{>0}$, the \textit{discrete} order of $p$ affects the likelihood. Thus, frequentist estimation is both a continuous and discrete problem.

The discrete aspect of estimation in the Mallows-Binomial model allows us to leverage existing algorithms from the Mallows model. As we will demonstrate, the inclusion of scores in the proposed model generally speeds up estimation as scores provide information on the strength of differences in object qualities, beyond their induced ranking. Still, exact computation of the MLE is difficult, or even intractable, as the number of objects increases. In this section, after some preliminaries, we propose exact and approximate algorithms to estimate the Mallows-Binomial MLEs. All algorithms are implemented in the \texttt{R} package \texttt{rankrate} \citep{rankrate}, which is publicly available on CRAN.

\subsection{Preliminaries}

Suppose $I$ judges assess a collection of $J$ objects using integer scores in the range $\{0,1,\dots,M\}$ and rankings of length $R$, such that $R\leq J$, where $M$, $J$, and $R$ are all known and fixed integers. We assume that each judge's ranking and scores are drawn independently from the same Mallows-Binomial($p,\theta$) distribution, where $p$ and $\theta$ are unknown and will be estimated via the method of maximum likelihood. Let $\pi_0=\text{Order}(p)$, $\Pi=\{\Pi_i\}_{i=1,\dots,I}$ denote the judges' rankings and $X=\{X_{ij}\}_{i=1,\dots,I}^{j=1,\dots,J}$ denote the judges' scores.

We begin by stating a useful property of the Kendall distance: For any two specific rankings $\pi_1,\pi_2$ of length $R$ and $J$, respectively, the Kendall distance can be written as 
\begin{align}
d_{R,J}(\pi_1,\pi_2) = \sum_{j=1}^R V_j(\pi_1,\pi_2),\label{Eq:KendallDecomp}
\end{align}
where $V_1(\pi_1,\pi_2)$ is the number of adjacency swaps needed to place the first object of $\pi_1$ in the first position of $\pi_2$, $V_2(\pi_1,\pi_2)$ is the number of additional adjacency swaps needed to place the second object of $\pi_1$ in the second position of $\pi_2$, and so on \citep{Fligner1986}. Note that each $V_j\in\{0,\dots,J-j\}$.

Then, the joint loglikelihood of the scores $X$ and rankings $\Pi$ is,

\begin{align}
&\ell(p,\theta|X=x,\Pi=\pi)\\
&= \log\prod_{i=1}^I\Bigg[\frac{e^{- \theta\sum_{j=1}^R V_j(\pi_i,\pi_0)}}{\psi_{R,J}(\theta)} \prod_{j=1}^J {M\choose x_{ij}}p_j^{x_{ij}}(1-p_j)^{M-x_{ij}}\Bigg]\nonumber \\
&= \sum_{i=1}^I \Bigg[ -\theta\sum_{j=1}^R V_j(\pi_i,\pi_0)-\log\psi_{R,J}(\theta)+\sum_{j=1}^J \Big[\log {M\choose x_{ij}}+x_{ij}\log p_j + (M-x_{ij})\log(1-p_j)\Big] \Bigg].
\end{align}
The maximum likelihood estimators, $(\hat p,\hat\theta)$, are therefore,

\begin{align}
    (\hat p,\hat\theta) &= \underset{p,\theta}{\arg\max} \sum_{i=1}^I \Bigg[ -\theta\sum_{j=1}^R V_j(\pi_i,\pi_0)-\log\psi_{R,J}(\theta)+\sum_{j=1}^J \Big[x_{ij}\log p_j + (M-x_{ij})\log(1-p_j)\Big] \Bigg]\nonumber\\
    &= \underset{p,\theta}{\arg\min} \Bigg\{ \theta\sum_{j=1}^R \overline{V}_j\Bigg\}+\Bigg\{\log\psi_{R,J}(\theta)\Bigg\}+\Bigg\{\sum_{j=1}^J \overline{x}_{j}\log \frac1{p_j} + (M-\overline{x}_{j})\log\frac1{1-p_j}\Bigg\} \nonumber \\
    &\equiv \underset{p,\theta}{\arg\min} \  f(p,\theta), \label{Eq:obj}
\end{align}
where $\overline{V}_j = I^{-1}\sum_{i=1}^I V_j(\pi_i,\pi_0)$ and $\overline{x}_j=I^{-1}\sum_{i=1}^I x_{ij}$.
As no analytic solution exists, the function $f$ within Equation \ref{Eq:obj} will be referred to interchangeably as a ``cost" or ``objective" function to be minimized via numerical optimization.

\subsection{Exact Algorithms based on A*}
\label{sect:exact}

The MLE $(\hat p,\hat \theta)$ induces an ordering of the true underlying object qualities, $\hat\pi_0=\text{Order}(\hat p)$. To find the MLE, we flip the problem around. Instead of optimizing over $p$ and $\theta$ directly, we first obtain $\hat\pi_0$ and then optimize for $\hat p$ and $\hat \theta$ under the constraints implied by $\hat\pi_0$ on $\hat p$.

\cite{Mandhani2009} and \cite{Meila2012} observed for the Mallows model that $\hat\pi_0$ could be estimated exactly using an A* algorithm. A* is a standard graph traversal algorithm developed by \cite{Hart1968}. Given a graph, A* finds the shortest path between a starting node and any terminal node. The algorithm requires a \textit{cost function} that measures the exact cost to get from the starting node to any other node, and a \textit{heuristic function} that estimates the remaining cost from any node to the nearest terminal node. The heuristic function is called \textit{admissible} when it guarantees a lower bound on the remaining cost. A* provably yields the shortest path when the heuristic is admissible. A trivial, admissible heuristic always returns 0, but results in an inefficient graph search. Oppositely, a maximal or near-maximal (``tight") admissible heuristic may reduce the number of nodes traversed during the search but be burdensome to compute and slow the overall algorithm.

A* algorithms traditionally define separate cost and heuristic functions but these functions are always used together \citep{Hart1968}. Thus, at each node the algorithm sums the cost and heuristic functions to lower bound the total cost possible given the current node. Due to the interdependent nature of the model parameters, we use an equivalent method of defining a single, admissible \textit{total cost heuristic} function which outputs a guaranteed lower bound on the total cost possible at any node in the graph. In other words, this single function is the sum of the usual cost and heuristic functions.

We propose two A* algorithms to calculate the exact MLE of the Mallows-Binomial model. Both algorithms use the same graph as in \cite{Mandhani2009} and \cite{Meila2012} but differ based on their admissible total cost heuristic functions; the first is crude but fast to compute, the second is tight but slow. We compare their overall speed in Section \ref{sect:simulation}.

\subsubsection*{Graph}

We define the graph $G$ as a tree that progressively adds one object to the ranking as you move down its branches. To specify a single starting node, we let the zero$^\text{th}$ layer of $G$ be empty. In the first layer, there is a node for each object in the collection. Traversing to any specific node in the first layer constrains the corresponding object to have the lowest-valued quality parameter (but does not specify any relationships among the remaining objects). For example, at node $n=(1)$ when $J=3$, the quality parameters are required to satisfy $p_1\leq p_2$ and $p_1\leq p_3$, but no relationship is specified between $p_2$ and $p_3$. Subsequent layers are successively formed from each node by adding a unique branch for each object not yet in the path to the node. Nodes in the $(J-1)^\text{th}$ layer are terminal as the last object is implied. For example, when $J=3$ the node $n=(3,2)$ is terminal as it implies the complete ordering of objects $(3,2,1)$. An example search graph when $J=3$ is shown in Figure \ref{fig:mandhani tree} (adapted from \cite{Mandhani2009}).

\begin{figure}[h!!]
    \centering
    \begin{tikzpicture}[level distance=1cm,
      level 1/.style={sibling distance=3cm},
      level 2/.style={sibling distance=1.5cm}]
      \node {$\emptyset$}
        child {node {1}
          child {node {1,2}}
          child {node {1,3}}
        }
        child {node {2}
          child {node {2,1}}
          child {node {2,3}}
        }
        child {node {3}
        child {node {3,1}}
          child {node {3,2}}
        };
    \end{tikzpicture}
    \caption{Graph for A* Search Algorithm with $J=3$ objects.}
    \label{fig:mandhani tree}
\end{figure}
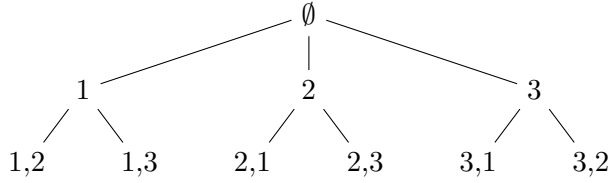

\subsubsection*{Crude Total Cost Heuristic}

Before stating our first total cost heuristic, we define a useful quantity based on rankings only: Let $Q$ be a $J\times J$ matrix such that each entry $Q_{uv}$, $u,v\in\{1,\dots, J\}$, is 
\begin{equation}
\label{Eq:Q}
    Q_{uv} = \frac{\sum_{i=1}^I I\{\text{object }u \text{ is ranked strictly higher than object } v \text{ in } \pi_i\}}{I}
\end{equation}

When $u=v$, it follows that $Q_{uv}=0$. If a comparison between objects cannot be deduced from any given ranking (due to partial rankings), we define the corresponding term in the numerator to be zero but do not change the denominator. Thus, $Q_{uv}+Q_{vu}=1$ whenever a strict ordering can be deduced between objects $u,v$ for all judges and is less than one otherwise. 
We are now ready to define the crude total cost heuristic.

\begin{definition}[Crude Total Cost Heuristic]\label{def:crude}

Let $n\in G$ such that $n=(n_1,\dots,n_k)$, $1\leq k\leq J-1$, where $n_1,\dots,n_k$ indicate unique objects in the collection $\{1,\dots,J\}$. Then, the crude total cost heuristic, $g_{c}(n): G \to \mathbb{R}$, is
\begin{align*}
    g_c(n) &=  \Big\{\hat\theta^n L\Big\} + \Big\{ \log\psi_{R,J}(\hat\theta^n)\Big\} + \Big\{\sum_{j=1}^J \overline{x}_{j}\log \frac{1}{\hat p^n_j} +(M-\overline{x}_{j})\log\frac{1}{1-\hat p^n_j}\Big\}\\
    L &= \Big(\sum_{\substack{v\in\{1:k\}\\u \in \{(v+1):J\}}}Q_{n_un_v}\Big) + \Big(\sum_{u,v\in\{(k+1):J\}} \min(Q_{n_un_v},Q_{n_vn_u})\Big)\\
    \hat\theta^n &= \underset{\theta}{\arg\min} \Big[\theta L +  \log\psi_{R,J}(\theta)\Big]\\
    \hat p^n &= \underset{p}{\arg\min} \Big[ \sum_{j=1}^J \overline{x}_{j}\log \frac{1}{p_j} +(M-\overline{x}_{j})\log\frac{1}{1- p_j}\Big]  \ \ \  \text{ s.t. } p_{n_1}\leq \dots \leq p_{n_k}, \  p_{n_k}\leq p_{n_l}, l> k.
\end{align*}
\end{definition}

The crude total cost heuristic may be seen as an extension of the quantity $L$ from \cite{Meila2012}. We prove that $g_c$ is admissible in Proposition \ref{prop:heur}.

\begin{proposition}
\label{prop:heur}
Under the conditions of Definition \ref{def:crude},
$$g_c(n) \leq \underset{p,\theta}{\arg\min} \ f(p,\theta) \ \ \text{ such that } \ \ p_{n_1}\leq \dots \leq p_{n_k}, \  p_{n_k}\leq p_{n_l}, l> k$$
and therefore $g_c(n)$ is admissible.
\end{proposition}
\begin{proof}
$g_c(n)$ consists of three terms which can each be mapped to a unique term in $f$. We prove the lower bound by proving (a) the first and second terms of $g$ are a lower bound on the corresponding terms in $f$, and (b) the third term of $g$ is a lower bound on the corresponding term in $f$.
\begin{itemize}
    \item[(a)] We first prove that $L\leq \sum_{j=1}^R \overline{V}_j$. Following closely the logic of \cite{Mandhani2009},
    \begin{align*}
        L &= \sum_{\substack{v\in\{1:k\}\\u \in \{(v+1):k\}}}Q_{n_un_v} + \sum_{u,v\in\{(k+1):J\}} \min(Q_{n_un_v},Q_{n_vn_u})\\
        &= \sum_{j\in\{1:k\}} \overline{V}_j + \sum_{u,v\in\{(k+1):J\}} \min(Q_{n_un_v},Q_{n_vn_u})\\
        &\leq \sum_{j\in\{1:k\}} \overline{V}_j + \sum_{j\in\{(k+1) : J\}} \overline{V}_j\\
        &= \sum_{j=1}^R \overline{V}_j
    \end{align*}
    The second line above holds by definition of $\overline{V}_j$ and the third line holds since one of $Q_{n_un_v},Q_{n_vn_u}$ must appear in the expression $\sum_{j\in\{(k+1):J\}}\overline{V}_j$. The fourth and final line holds since each $\overline{V}_j=0$ when $j>R$ definitionally. We complete (a) by again referencing \cite{Mandhani2009}, who proved that given $L$, $\hat\theta^n$ lower bounds the first two terms of $f$.
    \item[(b)] Since $\hat p^n$ is defined as the $\arg\min$ over $p$ for the third term of $f$ subject to the bare minimum constraints imposed by $n$, the third term of $g$ must lower bound the total cost. This is because as we traverse down the graph from $n$, only additional constraints may be imposed. Each additional constraint cannot lower the objective function, leading to a lower bound.
\end{itemize}
Therefore, $g_c(n)$ is an admissible total cost heuristic.
\end{proof}

Note that $g_c$ is suitably called crude because it is not necessarily a tight lower bound. Instead, the function independently lower bounds components of the likelihood corresponding to the Mallows and Binomial models. However, it is easy and quick to compute $L$ using matrix algebra, $\hat\theta^n$ via univariate optimization, and $\hat p^n$ via strictly convex optimization in a highly-constrained subspace of the $J$-dimensional unit hypercube.

\subsubsection*{LP Total Cost Heuristic}

In the crude total cost heuristic, it can be seen that the lower bound on the cost corresponding to the scores cannot be improved independently of the rankings, given $n$. A comparable statement is not true for the cost corresponding to rankings. The LP total cost heuristic makes the latter component tighter.

As a brief aside, the MLE of $\pi_0$ in the Mallows model is also the solution to the Kemeny ranking problem \citep{Meila2012}. \cite{Conitzer2006} proposed an algorithm to solve the Kemeny ranking problem based on an LP relaxation of the linear integer program that returns the minimum weight feedback edge set. Intuitively, the result can be understood as follows: In the crude lower bound, each pair of objects $u,v$ must be ranked such that $u$ is before $v$ or $v$ is before $u$. It does not take into account more complex relationships. For example, if $u$ is before $v$ and $v$ is before an object $w$, the lower bound would still illogically allow $w$ to be before $u$. The algorithm of \cite{Conitzer2006} removes this possibility. \cite{Mandhani2009} applied their result to an A* search algorithm for the Mallows model. In this paper, we extend this result to the Mallows-Binomial case.

\begin{definition}[LP Total Cost Heuristic]
\label{def:lp}
Let $n\in G$ such that $n=(n_1,\dots,n_k)$, $1\leq k\leq J-1$, where $n_1,\dots,n_k$ indicate unique objects in the collection $\{1,\dots,J\}$. Then, the LP Total Cost Heuristic, $g_{lp}(n): G \to \mathbb{R}$, is
\begin{align*}
    g_{lp}(n) &=  \Big\{\hat\theta^n L_{LP}\Big\} + \Big\{ \log\psi_{R,J}(\hat\theta^n)\Big\} + \Big\{\sum_{j=1}^J \overline{x}_{j}\log \frac{1}{\hat p^n_j} +(M-\overline{x}_{j})\log\frac{1}{1-\hat p^n_j}\Big\}\\
    L_{LP} &\text{ as defined in \cite{Conitzer2006}}\\
    \hat\theta^n &= \underset{\theta}{\arg\min} \Big[\theta L_{LP} +  \log\psi_{R,J}(\theta)\Big]\\
    \hat p^n &= \underset{p}{\arg\min} \Big[ \sum_{j=1}^J \overline{x}_{j}\log \frac{1}{p_j} +(M-\overline{x}_{j})\log\frac{1}{1- p_j}\Big]  \ \ \  \text{ s.t. } p_{n_1}\leq \dots \leq p_{n_k}, \  p_{n_k}\leq p_{n_l}, l> k
\end{align*}
\end{definition}

Note that $g_{lp}$ is identical to $g_c$ except for the replacement of $L$ with $L_{LP}$. We prove that $g_{lp}$ is a tighter lower bound than $g_c$ and admissible via Proposition \ref{prop:lpheur}.

\begin{proposition}
\label{prop:lpheur}
Under the conditions of Definition \ref{def:lp},
$$g_{c}(n) \leq g_{lp}(n)$$
for all nodes $n\in G$. Furthermore,
$$g_{lp}(n) \leq \underset{p,\theta}{\arg\min} f(p,\theta) \ \ \text{ such that } \ \ p_{n_1}\leq \dots \leq p_{n_k}, \  p_{n_k}\leq p_{n_l}, l> k$$
and therefore $g_{lp}(n)$ is admissible.
\end{proposition}

\begin{proof}
\cite{Conitzer2006} prove that $L\leq L_{LP}$. Note that $g_c$ and $g_{lp}$ are identical besides the replacement of $L$ with $L_{LP}$. Thus $g_{c}(x) \leq g_{lp}(x)$.

It was shown in \cite{Mandhani2009} that $L_{LP}\leq \sum_{j} \overline{V}_j$. In tandem with the proof of Proposition \ref{prop:heur}, $g_{lp}$ is admissible.
\end{proof}

\subsection{Approximate Algorithms}\label{sect:approx}

Exact algorithms to find the MLE of a Mallows model may be intractably slow when $J$ is large or consensus among judges is weak \citep{Mandhani2009}. To deal with such cases, a number of approximate search algorithms have been proposed \citep{Ali2012}. Here, we extend two simple, fast, and accurate algorithms proposed by \cite{Fligner1988} and \cite{Cohen1999}. We also state a third approximate algorithm which improves the accuracy of the latter algorithm at a computational cost. Each algorithm is described in turn.

\subsubsection*{FV Algorithm}

Under certain weak conditions, \cite{Fligner1988} found that the average ranking for each object is an unbiased estimator of the true consensus ranking in a Mallows model. The same paper proposed an approximate search algorithm for the MLE by averaging each object's rank position across judges and ordering the averages from best to worst into an ``average ranking". Then, one calculates the joint density of the data given the average ranking, as well as given each ranking one Kendall distance unit away from the average ranking. The ranking with the highest density in this small collection becomes the approximate MLE.
    
We propose a simple extension to the Mallows-Binomial model which we call ``FV". First, the algorithm calculates average rankings based on scores alone and rankings alone. If a distinct ordering of objects cannot be determined due to ties or partial rankings, all possible ways to break those ties are included in the set. Second, we calculate the joint density of the data given each of the average rankings and all rankings within one Kendall distance unit away from any of the average rankings. The ranking with the highest density becomes the approximate MLE, $\hat\pi_0$. Then, $\hat p$ and $\hat\theta$ are calculated conditional on $\hat\pi_0$.

\subsubsection*{Greedy Algorithm}

\cite{Cohen1999} proposed a greedy algorithm to approximate $\hat\pi_0$. Specifically, their algorithm iteratively estimates $\hat\pi_0$ by choosing the best available object at each ranking level from first to last. Here, ``best" means the object which least lowers the joint likelihood of the data given the current partial ordering. The algorithm is similar to the A* algorithms from Section \ref{sect:exact}, except there is no side-to-side traversal in the tree, e.g., once an object is selected for first place, that choice is never reconsidered. $\hat\theta$ is calculated conditional on $\hat\pi_0$. We can apply \citeauthor{Cohen1999}'s algorithm to the Mallows-Binomial model using the density function of the Mallows-Binomial model instead of the Mallows model, in what we call the ``Greedy" algorithm.

\subsubsection*{Greedy Local Algorithm}

The ``Greedy Local" algorithm extends the Greedy algorithm with a local search. Specifically, it runs the Greedy algorithm as stated and subsequently calculates the joint likelihood of the data given $\hat\pi_0$ and all rankings within one Kendall distance unit away from $\hat\pi_0$. If no ranking yields a higher likelihood than $\hat\pi_0$, the search stops. Else, $\hat\pi_0$ is updated to be the ranking with the current highest likelihood and the local search repeats until no better ranking is found. Then, $\hat p$ and $\hat\theta$ are calculated conditional on $\hat\pi_0$. 

The Greedy Local algorithm is slower than the Greedy algorithm, but guaranteed to estimate a $\hat\pi_0$ which yields a likelihood at least as great as that from the Greedy algorithm. When the Greedy algorithm identifies the exact MLE, the computational expense of performing the Greedy Local algorithm will be minimal as only one round of local search is performed.

\section{Algorithmic Speed and Accuracy Simulation Studies}\label{sect:simulation}

We now compare the speed and accuracy of estimation algorithms through a simulation study. We ran 20 unique simulations for each combination of model constants $I\in\{5,20,80\}$, $M\in\{10,20,40\}$, $J\in\{6,12,18\}$, and $R\in\{6,12,18|R\leq J\}$ and parameter $\theta\in\{1,2,3\}$. In each, we sampled $p$ randomly from a Uniform$[0,1]^J$. Then, estimation was performed on each data set using each of the 5 algorithms described in Section \ref{sect:estimation}: Crude, LP, FV, Greedy, and Greedy Local. The first two are exact algorithms while the latter three are approximate. We now demonstrate results separately based on speed and accuracy of the algorithms.

\subsection{Speed}

There are two useful metrics to consider when evaluating speed in graph search algorithms. The first is overall time, which we measure in seconds. The second is the number of nodes traversed, which signifies an efficient algorithm with respect to memory. While time may be a more practically important metric, if the number of nodes traversed is substantially smaller for a slower algorithm then potential improvements to memory time or code efficiency may ultimately result in a faster algorithm. We compare algorithm speed in Figure \ref{fig:speed} on the basis of these two metrics.

\begin{figure}[h!!]
    \centering
    \includegraphics[width=\textwidth]{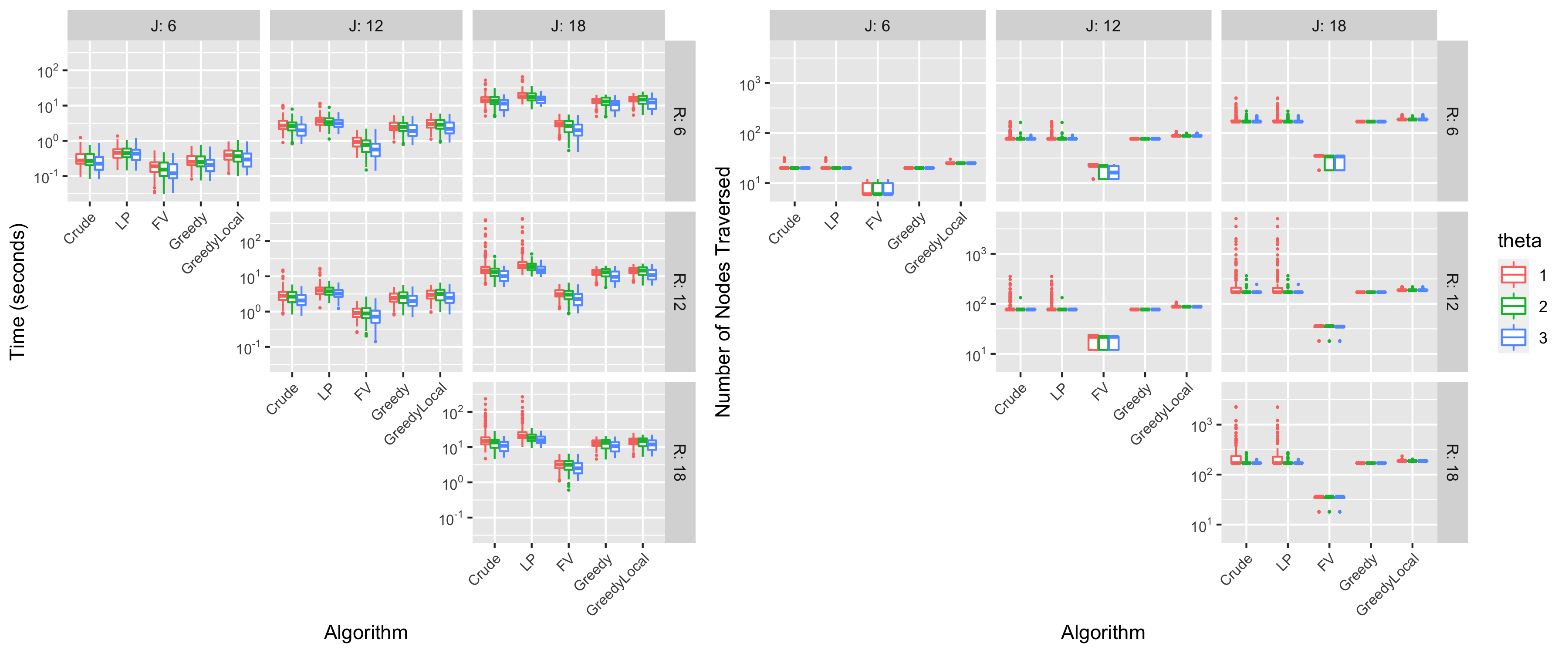}
    \caption{Speed of algorithms based on time (\textit{left}) and number of nodes traversed (\textit{right}) across different values of $J$, $R$, and $\theta$. Results are aggregated over $M$ and $I$.}
    \label{fig:speed}
\end{figure}

Among the exact algorithms, the number of nodes traversed is comparable between both yet computation time is reasonably higher for LP under most regimes. When exact search is desired, we recommend the Crude algorithm on the basis of these results. Regarding the approximate algorithms, the FV algorithm is substantially faster than the rest. This difference is by approximately an order of magnitude in all regimes. The Greedy and Greedy Local algorithms are generally similar in speed to the Crude exact algorithm, a potentially disappointing result given that Greedy and Greedy Local operate under no guarantee of providing an exact solution. However, they exhibit consistent speed results, unlike the exact algorithms which have frequent and extreme high-time outliers. 

Overall, we observe that estimation time generally increases as $J$ increases and decreases as $\theta$ increases. These results should not be surprising: For large $J$, the algorithms can be slow due to the massive parameter domain. When $\theta$ is small, ranking consensus is weak so search algorithms may be pulled into many distinct subspaces of the parameter domain. Speed does not change substantially as $R$ increases. 

\subsection{Accuracy}

We measure accuracy of the approximate search algorithms using two metrics: The first is the proportion of simulations in which each algorithm returns the true MLE. However, incorrect estimates may be trivially different from the truth, which leads us to our second metric: The Kendall distance to the true MLE. This measures how far away the estimated ordering of the object quality parameters are from $\hat\pi_0$. We compare algorithm accuracy in Figure \ref{fig:accuracy} on the basis of these two metrics.

\begin{figure}[h!!]
    \centering
    \includegraphics[width=\textwidth]{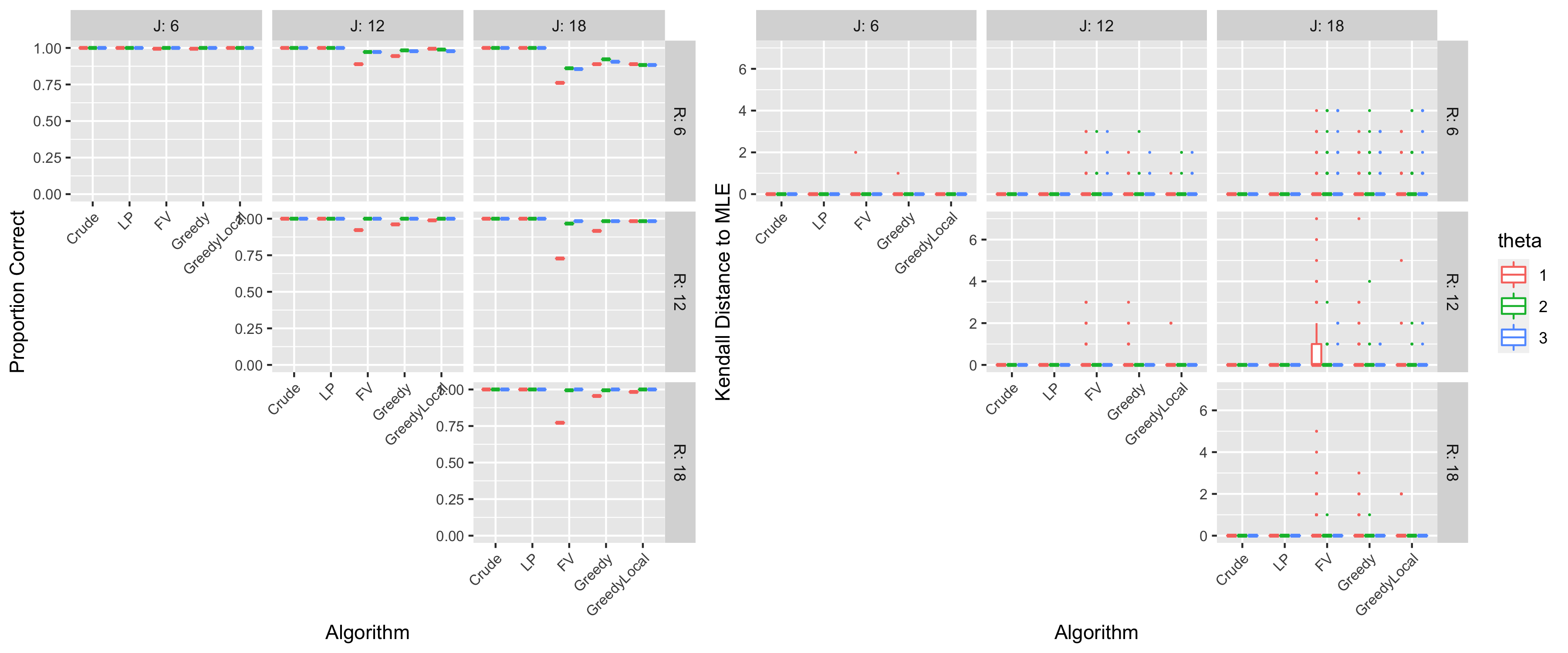}
    \caption{Accuracy of algorithms based on the proportion of estimates equal to the true MLE (\textit{left}) and Kendall distance to the true MLE (\textit{right}) across different values of $J$, $R$, and $\theta$. Results are aggregated over $M$ and $I$.}
    \label{fig:accuracy}
\end{figure}

The proportion correct will be 1 and the Kendall distance to the true MLE will be 0 for both the exact algorithms, by definition. For the approximate algorithms, both metrics suggest the order of least to most accurate approximate algorithm is FV, Greedy, and Greedy Local. We point out that even though FV was the fastest algorithm, it exhibits the worst accuracy overall, especially when $\theta$ is small. On the other hand, Greedy Local is quite often exactly correct. Accuracy generally improves in all approximate algorithms as $R$ increases, which makes sense given that partial rankings equate to less preference information. 


In sum, this section has provided insights for practitioners when selecting an estimation algorithm for the Mallows-Binomial model. If exact MLEs are desired, the Crude algorithm is a good choice. When approximations are satisfactory or required due to computational cost, especially when $J$ is large or postulated $\theta$ is small, we recommend the Greedy Local algorithm due to its high accuracy or the FV algorithm for a fast and rough approximation of the consensus ranking.

\section{Analysis of Grant Panel Review Data}\label{sect:realdata}

We now apply our model to a real data set on grant panel review. After providing an exploratory analysis, we display and interpret estimation results.

\subsection{Exploratory Analysis}

We consider one specific instance of grant panel review conducted by the American Institute of Biological Sciences (AIBS) during Fall 2020, where judges provided both scores and rankings. In the panel, 9 judges discussed 18 proposals. They were allowed to assign scores between 1.0 and 5.0 in single decimal point increments. After discussion and openly scoring each proposal in turn, judges were asked to provide top-6 partial rankings in private. Ties were not allowed. Since judges discussed every proposal, a proposal not receiving a top-6 ranking was deemed worse than each of the ranked top-6 proposals. With a few exceptions, all judges scored all proposals and ranked their top 6. One judge scored only one proposal and did not provide a ranking; another did not provide a ranking, and a third only provided a top-5 ranking. Based on information from the AIBS, missing data occurred for reasons independent of any characteristics of the proposals, such as child care or family responsibilities as panel review discussions occurred remotely during the Covid-19 pandemic. Thus, we can assume the missing data to be missing completely at random. In this case, estimation using all available (partial or complete) rankings and scores will not be biased \citep{little2019statistical}.  If missingness was due to circumstances related to object quality, for example, one would have to carry out a different treatment of missing data \citep{little2019statistical}. Figure \ref{fig:panel1_EDA} summarizes the data with proposal scores on the left and partial rankings on the right.

\begin{figure}[h!!]
    \centering
    \includegraphics[width=\textwidth]{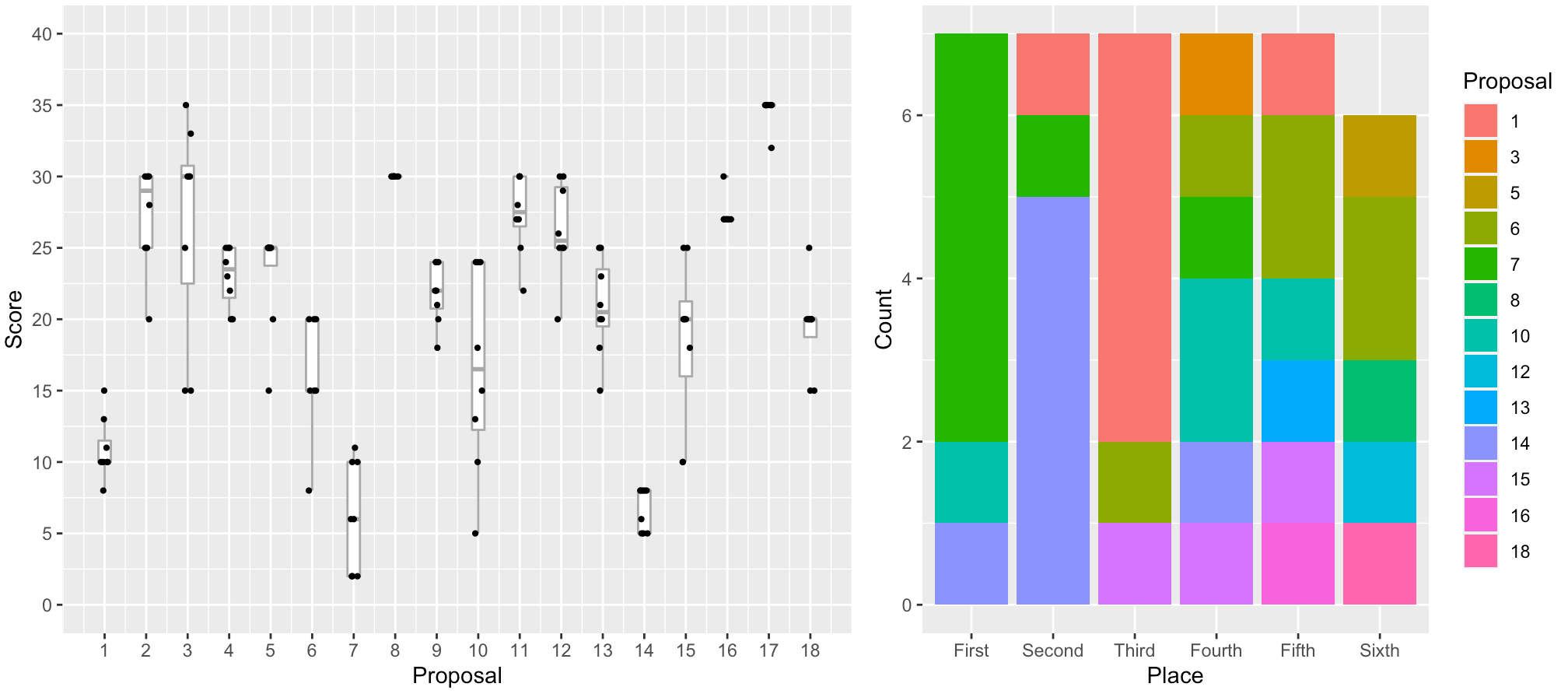}
    \caption{AIBS grant panel review data. \textit{Left}: Scores by proposal (boxplots in gray; raw data in black). \textit{Right}: Proposals by rank.}
    \label{fig:panel1_EDA}
\end{figure}

We observe a variety of scoring and ranking patterns by proposal. For some proposals all judges gave identical scores, while for others there was wide disagreement among judges. Additional calculations (shown in Appendix \ref{sect:A_mv}) indicate that sample variances of scores for each proposal are roughly centered around the theoretical variances based on Binomial score models, and differences may be partially explained by the small sample size. Furthermore, proposals with moderate scores tend to have higher variances than those with generally high or low scores. These observations suggest that Binomial score models are reasonable for this data. For rankings, 13 of the 18 proposals were in at least one judge's top-6 ranking. However, Figure \ref{fig:panel1_EDA} shows that a smaller subset of proposals were ranked by a majority of the judges (e.g., proposals 1, 7, and 14). Separately, we also measure the consistency between rankings and scores at the judge level. If the rankings were to always align with the order of the scores, for example, the rankings may be thought of as providing little additional information. To quantify this, we measure the Kendall distance (i.e., the number of pairwise disagreements) between each judge's partial ranking and the implied order of his/her scores. When a judge assigns equal scores to any two proposals or does not rank any two proposals, we do not count potential inconsistencies between them. We found the Kendall distances between each judge's partial ranking and score-implied ranking to be $\{2,4,4,5,7,11,22\}$ (ordered from least to greatest). Given that each judge only provided a top-6 ranking of the proposals, there is substantial discordance between rankings and scores at the judge level. We believe this further motivates the use of a combined model for rankings and scores for this data set.

The AIBS is principally interested in identifying which proposals should receive funding. While thematic and other considerations also contribute to funding decisions, funding agencies rely on peer review to identify which proposals are quality proposals and whether proposals can be ordered or tied in quality. 
Thus, both estimating proposal quality parameters and identifying a consensus ranking are of interest.
Understanding uncertainty in the estimated consensus ranking is key for understanding if objects are of similar quality. 

We fit a Mallows-Binomial model to the data, in which $M=40$, $I=9$, $J=18$, and $R=6$. In doing so, we make note of a few assumptions. First, we assume that each proposal has a true underlying quality. The underlying qualities imply a true ordering of the proposals from best to worst, which we seek to estimate. Second, we assume that the population of judges is homogeneous in its preferences. This may be interpreted as assuming that all judges use the same criteria when ranking or scoring and that all variation in scores and rankings is due to random chance, as opposed to true ideological differences. Third, we assume that all scores and rankings are conditionally independent given the latent true underlying quality of a proposal and the level of consensus strength.

\subsection{Results}

We now present the MLE and the associated bootstrapped 90\% confidence intervals of the consensus scale parameter $\theta$ and object quality vector $p$. Confidence intervals are based on 200 bootstrap samples. Table \ref{tab:MLE} contains parameter estimates and Figure \ref{fig:results_mle} displays estimates alongside the data. In the left panel, expected scores and associated confidence intervals overlay judges' observed scores. The right panel is a histogram of the Kendall distance between each judge's partial ranking and the estimated MLE of the consensus ranking.

\begin{table}[ht]
\centering
\begin{tabular}{ccc|ccc}
  \hline
Parameter & MLE & 90\% CI & Parameter & MLE & 90\% CI\\ 
  \hline
$\theta$ & 0.529 & (0.421,1.124) & $p_{10}$ & 0.416 & (0.308,0.525) \\ 
  $p_1$ & 0.272 & (0.239,0.306) & $p_{11}$ & 0.684 & (0.642,0.730) \\ 
  $p_2$ & 0.683 & (0.626,0.729) & $p_{12}$ & 0.656 & (0.565,0.711) \\ 
  $p_3$ & 0.666 & (0.544,0.766) & $p_{13}$ & 0.522 & (0.481,0.565) \\ 
  $p_4$ & 0.575 & (0.553,0.616) & $p_{14}$ & 0.169 & (0.150,0.186) \\ 
  $p_5$ & 0.563 & (0.511,0.646) & $p_{15}$ & 0.463 & (0.374,0.541) \\ 
  $p_6$ & 0.400 & (0.325,0.453) & $p_{16}$ & 0.683 & (0.646,0.698) \\ 
  $p_7$ & 0.153 & (0.103,0.199) & $p_{17}$ & 0.866 & (0.850,0.875) \\ 
  $p_8$ & 0.750 & (0.711,0.750) & $p_{18}$ & 0.484 & (0.444,0.541) \\ 
  $p_9$ & 0.563 & (0.526,0.588) &  &  &  \\ 
   \hline
\end{tabular}\\
$\hat\pi_0 = \{7, 14, 1, 6, 10, 15, 18, 13, 5, 9, 4, 12, 3, 16, 2, 11, 8, 17\}$
\caption{Maximum likelihood estimates of model parameters for the AIBS grant panel review data.\label{tab:MLE}}
\end{table}

\begin{figure}[h!!]
    \centering
    \includegraphics[width=.9\textwidth]{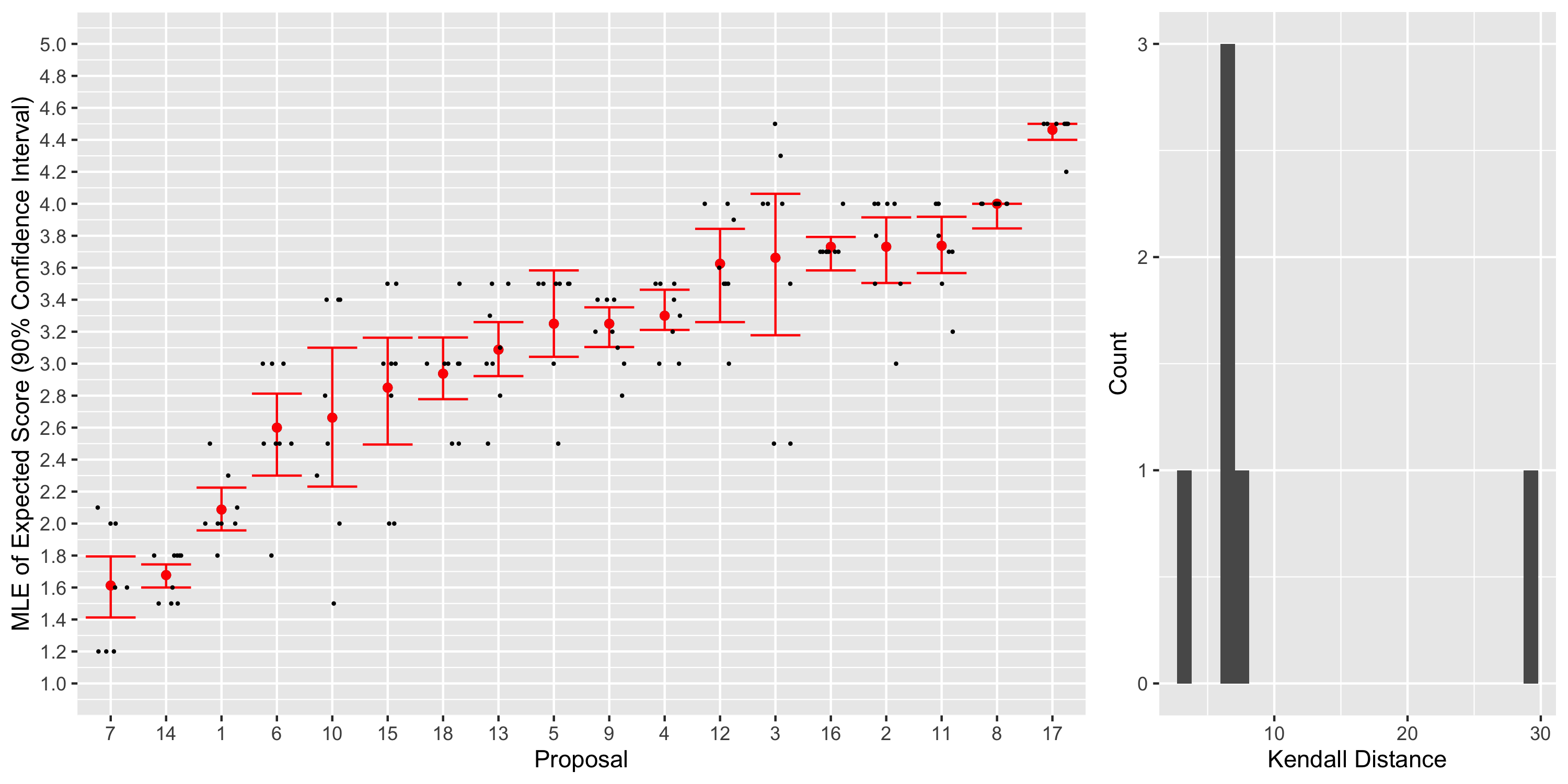}
    \caption{Results overlaid on data. \textit{Left:} Scores by proposal (black) and MLE and 90\% confidence intervals of expected score (red). Expected and observed scores are provided on their original scale. The order of proposals on the x-axis aligns with the MLE of the consensus ranking in the Mallows-Binomial model. \textit{Right:} Histogram of Kendall distance $d_{R,J}$ between each judge's partial ranking and $\hat\pi_0$.}
    \label{fig:results_mle}
\end{figure}

As shown in the left panel of Figure \ref{fig:results_mle}, the MLEs of the expected scores approximately equal the mean observed scores. However, confidence bands reflect information obtained from both scores and rankings. For example, proposals 8 and 16 have lower confidence limits that are much lower than the minimum score they received, which is unusual for a measure of the expected (mean) score. This likely occurs since they were each ranked comparatively better than the scores they received on average. We also notice that a few proposals share the same MLE of true underlying quality but are strictly ordered (i.e., not tied) in the consensus ranking. For example, proposals 5 and 9 correspond to $\hat p_5=\hat p_9 = 0.563$, but proposal 5 is ranked higher than proposal 9 in $\hat\pi_0$. In this case, proposal 5 received a marginally worse average score than 9 but was ranked higher. Thus, the model can capture a difference in ranking while suggesting the true underlying quality is likely nearly identical. In the right panel, we notice that most judge's partial ranking mostly aligned with the MLE of the consensus ranking. The outlier in the right panel of Figure \ref{fig:results_mle} corresponds to a judge who assigned top-6 rankings to three proposals with comparatively poor scores (proposals 3, 8, and 16).

We display the estimated consensus ranking and associated 90\% ranking confidence intervals for each proposal based on the Mallows-Binomial model in Figure \ref{fig:ER}. Additionally, we show results that would be obtained under four separate ranking aggregation or score aggregation models. The first model, \textit{Converted Scores}, uses scores and rankings converted into scores for each judge such that the first-ranked object receives that judge's best score, the second-ranked object receives that judge's second-best score, etc., as suggested by \cite{Li2009}. Then, the model uses independent Binomial score distributions for each proposal (no rankings are modeled). The second comparison model, \textit{Only Scores}, is identical to the first but excludes all rankings. The third comparison model, \textit{Converted Rankings}, uses rankings and scores converted into rankings for each judge by simple ordering (ties are broken at random). Then, a Mallows distribution is used to model the ranking data. The fourth comparison model, \textit{Only Rankings}, is identical to the third but excludes all scores. We note that \textit{Converted Scores} and \textit{Converted Rankings} use all the available data (after conversion) and therefore provide the most direct comparison to the Mallows-Binomial, while \textit{Only Scores} and \textit{Only Rankings} are limited by the exclusion of certain preference data; none of the comparison methods jointly model the original rankings and scores. Confidence intervals for each model are based on 200 bootstrap samples.

\begin{figure}[h!!]
    \centering
    \includegraphics[width=1\textwidth]{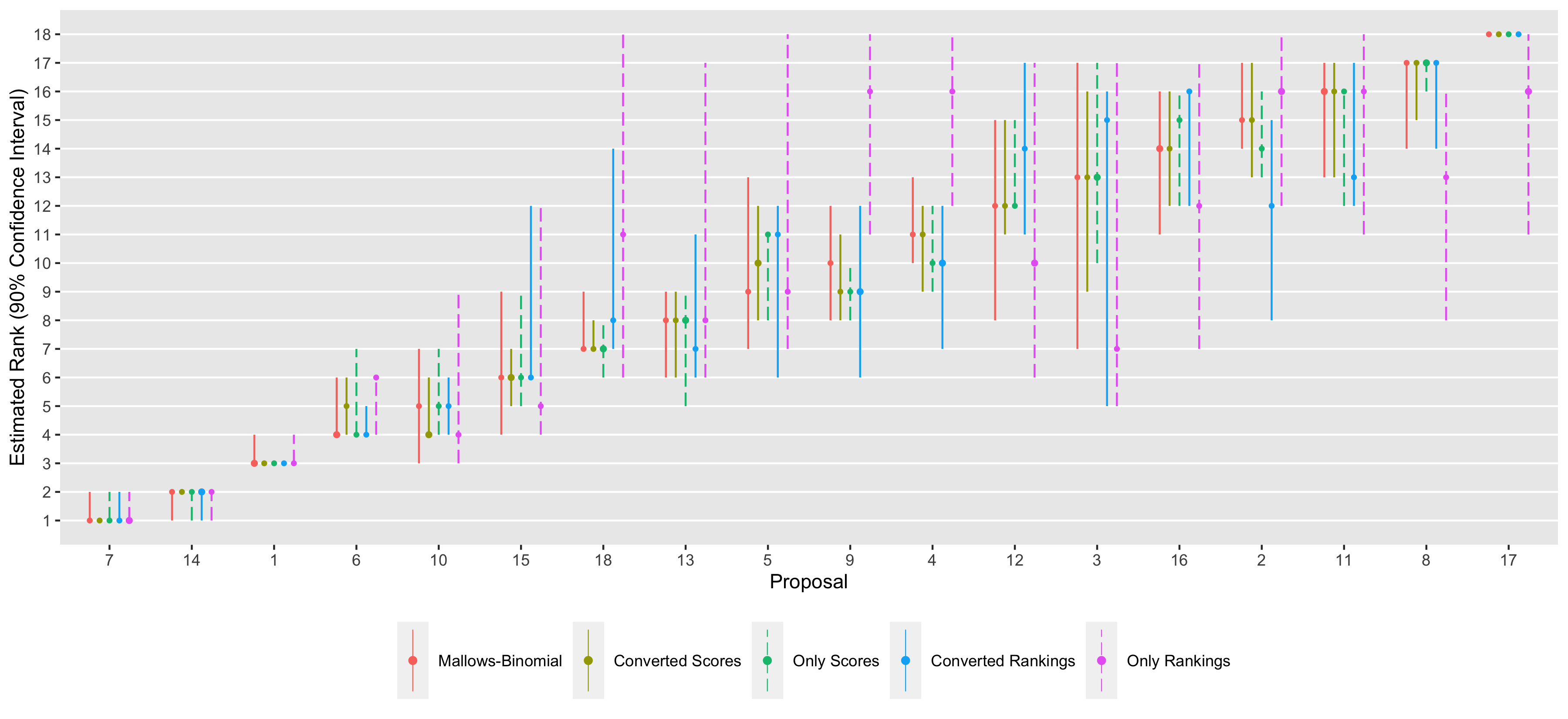}
    \caption{Estimated ranks and 90\% confidence intervals for the \textit{Mallows-Binomial} model based on scores and partial rankings and four competing models based on: (1) scores and rankings converted into scores in Binomial models (\textit{Converted Scores}), (2) scores in Binomial models that exclude rankings (\textit{Only Scores}), (3) partial rankings and scores converted into rankings in a Mallows model (\textit{Converted Rankings}), and (4) rankings in a Mallows model that excludes scores (\textit{Only Rankings}). Confidence intervals for models that exclude data (either rankings or scores, respectively) are represented with dashed lines. The order of proposals on the x-axis aligns with the MLE of the consensus ranking in the \textit{Mallows-Binomial} model.}
    \label{fig:ER}
\end{figure}

We observe in Figure \ref{fig:ER} that the Mallows-Binomial model provides a sensible estimated ranking for each proposal: Each proposal has a unique point estimate for rank place and the associated 90\% confidence intervals reflect the scores and ranks it received. For example, proposal 7 was ranked first by 5 of the 7 judges and had the best average score, but proposal 14 was highly ranked by many judges and received a similarly high average score. Thus, the 90\% confidence intervals of (1,2) for the rank place of proposals 7 and 14 appear appropriate. On the other hand, proposal 3 received the 13th best average score, which corresponds to its point estimate for rank place. However, its 90\% confidence interval (7,17) for rank place is appropriately wide given its wide range of scores (minimum 15, maximum 35) and a single fourth-place ranking, which injects uncertainty into the model. In general, confidence intervals are narrow when consensus between scores and rankings across judges is strong and are wider otherwise.

Results from the Mallows-Binomial model improve upon results from the other models in unique ways. The \textit{Converted Scores} and \textit{Only Scores} models provide similar rank place point estimates to the Mallows-Binomial model but confidence intervals that may be considered inappropriate. \textit{Converted Scores} provide narrow intervals that reflect an artificially inflated sample size (resulting from combining both original and converted scores) but does not account for uncertainty arising from converting rankings into scores. Using only scores limits the amount of information on judges' perception of proposal quality via rankings, which naturally leads to a loss in precision. However, sometimes the \textit{Only Scores} model exhibits narrower confidence intervals than the Mallows-Binomial model when scores are consistent but rankings are not, which still falsely reflects the true combined preferences of the judges. The \textit{Converted Scores} and \textit{Only Scores} models do not estimate the consensus scale parameter $\theta$.

Point estimates and confidence intervals from the \textit{Converted Rankings} and \textit{Only Rankings} models differ substantially from those of the Mallows-Binomial. Differences are particularly apparent for proposals ranked in 7th place or worse, as those proposals generally have less data due to the partial rankings collected. The \textit{Converted Rankings} model loses precision compared to the Mallows-Binomial model in the top ranking places, despite having the same number of observations, since scores converted into rankings via ordering lack information on the strength of the difference in quality between proposals. The \textit{Only Rankings} model has even less precision, since the complete exclusion of scores and limited information provided by partial rankings constrains inference on the many proposals that were never or rarely ranked and leads to uninformative and insensible rankings. For example, proposals 2, 4, 9, 11, and 17 have near-identical and wide confidence bands as they were never ranked, while proposals 3, 5, 8, 12, 13, 16, and 18 have even wider confidence bands since they were ranked only by a small number of judges. Furthermore, the \textit{Converted Rankings} and \textit{Only Rankings} models do not estimate the object quality parameter vector $p$.

Results from the Mallows-Binomial model allow us to compare proposals with confidence. For example, the model suggests that proposals 7 and 14 are of similarly high quality, but that relative quality is harder to differentiate for proposals 1, 6, and 10. These types of comparisons may be useful when drawing a funding line at the AIBS. If the AIBS can fund, for example, only 6 proposals, then using 90\% marginal confidence intervals by proposal they should fund proposals 7, 14, 1, 6 and select two additional proposals between  10, 15, and 13 (perhaps based on point estimates or a random lottery).

\section{Discussion}\label{sect:discussion}

In this paper, we proposed the first unified statistical model for rankings and scores that does not involve data conversion, the Mallows-Binomial model. We formulated a computationally efficient algorithm to find the exact maximum likelihood estimators of model parameters and demonstrated statistical properties of the model such as bias, consistency, and variance of estimators. This research aligns well with the recommendations from a peer review study at the 2016 Neural Information Processing Systems conference that recommended using both rankings and scores to gain benefits from each data format~\citep{shah2018design}. That study also emphasized the need to design algorithms to efficiently combine scores and rankings for further guidance on conference submission quality~\citep[][p.27]{shah2018design}. 

We applied the Mallows-Binomial model to grant review data which collected both scores and partial rankings from a panel of judges. The model was used to identify a consensus ranking based on the scores and partial rankings. The estimated consensus ranking was different from what would be obtained with comparable models for (converted) scores or rankings alone. Furthermore, we demonstrated a method to obtain confidence bands of proposal qualities and/or rank places via the bootstrap that can be used to select proposals that are preferred by reviewers  with statistical confidence. Confidence bands clearly reflect information from both scores and rankings provided by the judges. 

The proposed model is useful whenever both rankings and scores for a collection of objects are available. Beyond the example presented here, this may occur in a variety of contexts. For example, relevance of webpages to a search query may be measured from different systems using either numerical metrics (scores) or ordinal comparisons (rankings). In this example, the object quality parameters would measure both relative and absolute relevance to the search query, and the scale parameter would represent consensus among judges. If scores and rankings arise from the same system, rankings may help break ties when scores are close; if different systems are used to provide different scores and rankings, using all available data increases estimation precision. In contrast to methods that convert rankings and scores into data of a single type, the proposed model removes the potential introduction of error by using information from both sources directly. Yet, it allows for using both rankings and scores to express different types of comparison and levels of granularity in preferences. Furthermore, because both types of data are incorporated in a statistical model, this allows for uncertainty quantification in the estimation of true underlying quality, consensus, and strength of consensus when both scores and rankings are present using standard model-based statistical approaches.

Estimation methods presented in this paper for the Mallows-Binomial model can be improved or extended upon in a number of ways. Computational efficiency of estimation may be improved via a different heuristic function in the A* algorithm and permit exact estimation of the model in the presence of large numbers of objects. Approximate algorithms may be improved to increase accuracy and/or speed. In addition, alternative Bayesian estimation methods may be developed by extending the work of \cite{Vitelli2018} on the Mallows model. Model components may also be generalized: For example, the Generalized Mallows distribution or Infinite Generalized Mallows distribution proposed by \cite{Meila2012} may replace the ranking distribution component of our proposed model. The Bradley-Terry or Plackett-Luce distributions may also replace the ranking distribution component and allow for additional types of ranking data (such as pairwise or group-wise comparisons beyond top-$R$ partial rankings), although the downside is that these distributions do not directly include a consensus ranking parameter. Additionally, the Beta-Binomial or Poisson distributions may replace the Binomial score distribution component in our proposed model if one was interested in accounting for differences in the variance of object scores among judges or working with score data with no theoretical maximum value, respectively. Lastly, the model may be considered in a latent class framework to identify the presence of and to measure local consensus among heterogeneous preference groups, e.g., by extending earlier work on the mixture of Mallows distributions \citep{Busse2007} or  Plackett-Luce distributions \citep{Gormley2006,Gormley2009}.

\section{Acknowledgements}
The authors thank Stephen Gallo and the American Institute of Biological Sciences for the data used in this article. This work was funded by NSF grant 2019901.

\newpage
\appendix

\section{Bias of $(\hat p,\hat\theta)$ in the Mallows-Binomial distribution}\label{sect:A_bias}

We are interested in estimating the MLEs $(\hat p,\hat\theta)$ using a collection of $I$ independent and identically distributed samples from a Mallows-Binomial($p,\theta$) distribution, \\$(X,\Pi)=\{(X_1,\Pi_1), \dots, (X_I,\Pi_I)\}$. 

\subsection*{Example Demonstrating Bias}

Considering the following example: Let $I=1$, $J=R=3$, and $M=1$. Suppose $p_0=[0.1,0.4,0.9]^T$ and $\theta_0 = 1$. Then, there are 48 unique, possible observations. Specifically, they are the combinations of $X_1\in\{0,1\}$, $X_2\in\{0,1\}$, $X_3\in\{0,1\}$, and $\Pi\in\{123,132,213,213,312,321\}$.

We do not enumerate all data combinations or their associated likelihoods and MLEs. However, we state the bias of the MLEs below:
\begin{align*}
    \text{Bias}(\hat p_1) &= E_{P_{\theta_0,p_0}}[\hat p_1] - p_{01} \approx 0.1419 -  0.1 = 0.0419\\
    \text{Bias}(\hat p_2) &= E_{P_{\theta_0,p_0}}[\hat p_2] - p_{02} \approx 0.4192 -  0.4 = 0.0192\\
    \text{Bias}(\hat p_3) &= E_{P_{\theta_0,p_0}}[\hat p_3] - p_{03} \approx 0.8390 - 0.9 = -0.0610\\
    \text{Bias}(\hat \theta) &= E_{P_{\theta_0,p_0}}[\hat \theta] - \theta_{0} = \infty - 1 = \infty
\end{align*}
Thus, all MLEs in this example are biased. Code to replicate this simulation can be found in the code repository.

\subsection*{Intuition for Bias}

We now provide intuition for why bias exists in the Mallows-Binomial model. We start by proving that $(\hat p,\hat\theta)$ is biased when $\pi_0$ is known: Assume $\pi_0$ is fixed and known. Then,

\begin{align*}
    (\hat p,\hat\theta)|\pi_0 &= \underset{p,\theta|\text{Order}(p)=\pi_0}{\arg\max}\prod_{i=1}^I \frac{e^{-\theta d(\pi_i,\pi_0)}}{\psi_R(\theta)}\prod_{j=1}^J p_j^{x_{ij}}(1-p_j)^{M-x_{ij}}\\
    &= \underset{p,\theta|\text{Order}(p)=\pi_0}{\arg\max} -\theta\Big[\sum_{i=1}^I d(\pi_i,\pi_0)\Big]-I\log\psi_R(\theta) + \\
    & \ \ \ \ \ \ \ \ \ \ \ \ \ \ \ \ \ \ \ \ \ \ \ \sum_{j=1}^J\Big[\sum_{i=1}^I x_{ij}\log p_j + (IM-\sum_{i=1}^I x_{ij})\log(1-p_j)\Big]
\end{align*}

In the above expression, $p$ and $\theta$ factor. Thus,

\begin{align*}
\hat \theta|\pi_0 &= \underset{\theta|\pi_0}{\arg\max} -\theta\Big[\sum_{i=1}^I d(\pi_i,\pi_0)\Big]-I\log\psi_R(\theta)\\
\hat p|\pi_0 &= \underset{p|\text{Order}(p)=\pi_0}{\arg\max}  \sum_{j=1}^J\Big[\overline{x}_{j}\log p_j + (M-\overline{x}_{j})\log(1-p_j)\Big]
\end{align*}
\cite{Tang2019} proved that $E[\hat\theta|\pi_0] > \theta$. So when $\pi_0$ is known, $\hat\theta$ is biased upward.

Regarding $\hat p$, the problem is now precisely an \textit{isotonic regression} problem for binomial parameters. It was shown in \cite{ayer1955empirical} and \cite{barlow1972isotonic} that
$$
\hat p_j|\pi_0=
\begin{cases}
\overline{x}_j/M & \overline{x}_1,\dots,\overline{x}_{j-1}\leq \overline{x}_j\leq \overline{x}_{j+1},\dots,\overline{x}_{J}\\
(\sum_{i\in A_j}\overline{x}_i)/M, & \text{otherwise}
\end{cases}
$$
where $A_j$ is the intersection of the lower and upper sets of $\pi_0$ that include $j$. \cite{robertson1988order} proved that $E[\hat p_j|\pi_0] \neq p_j$ in generality, i.e., the parameters $p_j$, $j=1,\dots,J$ are biased. The direction of the bias may vary between each $p_j$.

Thus, $(\hat p,\hat\theta)|\pi_0$ is biased. We would now like to prove that the MLEs are biased even when $\pi_0$ is not known. This is a substantially more challenging problem due to the interconnectedness of $\hat p$ and $\hat\theta$ during estimation. A complete proof is left open. However, the previous counterexample demonstrates that bias is present in at least some situations while simulations in Section \ref{sect:properties} demonstrate the bias is often minimal.

\section{Consistency of $(\hat p,\hat\theta)$ in the Mallows-Binomial Distribution}\label{sect:A_consist}

Let $(X,\Pi)_I$ denote a collection of $I$ independent and identically distributed observations from a Mallows-Binomial($p_0,\theta_0$) distribution, where $(p_0,\theta_0)$ is unknown and will be estimated via maximum likelihood. Assume that each true $p_{0j}$, $j=1,\dots,J$ lie within in the unit interval and are each bounded away from 0 and 1, and that the true $\theta_0$ is less than $J$ and is bounded greater than 0. We note that the restrictions on $p_{0}$ ensure each proposal may receive any integer score between 0 and M with positive probability. Furthermore, the restrictions on $\theta_0$ ensure there is not a complete lack of consensus (as if $\theta_0$ were equal to 0) and does not substantially impact situations of near-complete consensus (when $\theta_0\geq J$, it is near certain that all rankings are identical and match the true $\pi_0$ regardless of how large $\theta_0$ is). Under these assumptions, the unknown parameters live in a compact space. We denote this space by $\Theta$.

The loglikelihood (less a normalizing constant) of each observation $(X_i,\Pi_i)$, $i=1,\dots,I$, is
$$\ell_{p,\theta}(X_i,\Pi_i) = -\theta d_{R,J}(\Pi_i,\text{Order}(p)) - \log\psi_{R,J}(\theta) + \sum_{j=1}^J [X_{ij}\log p_j + (M-X_{ij})\log(1-p_j)]$$
It can be seen that $d_{R,J}(\Pi_i,\text{Order}(p))\in\{0,1,\dots,J(J-1)/2\}$ and that $\psi_{R,J}(\theta)\in(1,J!)$. The loglikelihood is not continuous in $p$. Discontinuities may exist when $p_j=p_k$, $j\neq k$.

We would like to show the consistency of the maximum likelihood estimators (MLE), $(\hat p_I,\hat\theta_I)$ to the true parameters, $(p_0,\theta_0)$. We define a few quantities:

\begin{align*}
M_I(p,\theta) &= \frac1I\sum_{i=1}^I \ell_{p,\theta}(X_i,\Pi_i)\\
    &= -\theta \frac1I\sum_{i=1}^I d_{R,J}(\Pi_i,\text{Order}(p)) -\log\psi_{R,J}(\theta) + \sum_{j=1}^J [\frac1I\sum_{i=1}^I X_{ij}\log p_j + (M-\frac1I\sum_{i=1}^I X_{ij})\log(1-p_j)]\\
    &= -\theta \overline{d}_{R,J}(\Pi,\text{Order}(p)) -\log\psi_{R,J}(\theta) + \sum_{j=1}^J[ \overline{X}_{j}\log p_j + (M-\overline{X}_j)\log(1-p_j)]\\
M(p,\theta) &= E[\ell_{p,\theta}(X_1,\Pi_1)]\\
&= -\theta\Big[\frac{Re^{-\theta}}{1-e^{-\theta}} - \sum_{j=J-R+1}^J \frac{je^{-j\theta}}{1-e^{-j\theta}}\Big] -\log\psi_{R,J}(\theta) + \sum_{j=1}^J\Big[ Mp_j\log p_j + (M-Mp_j)\log(1-p_j)\Big],
\end{align*}

since
\begin{align*}
    E[X_j] &= Mp_j, \ \ \forall j\\
    Var[X_j]&=Mp_j(1-p_j) \ \ \forall j\\
    E[d_{R,J}] &= \frac{Re^{-\theta}}{1-e^{-\theta}} - \sum_{j=J-R+1}^J \frac{je^{-j\theta}}{1-e^{-j\theta}}\\
    Var[d_{R,J}] &= \frac{Re^{-\theta}}{(1-e^{-\theta})^2} - \sum_{j=J-R+1}^J \frac{j^2e^{-j\theta}}{(1-e^{-j\theta})^2}.
\end{align*}
The first and second lines are standard results regarding Binomial random variables and the third and fourth lines follow directly from  \cite{Fligner1986}.

We now are ready to prove the consistency of the maximum likelihood estimators. We do so using Theorem 5.7 of \cite{van2000asymptotic}. Specifically, we must prove (1) uniform consistency of $M_I(p,\theta)$ to $M(p,\theta)$, (2) the true parameter values are well-separated in $M(p,\theta)$, and (3) that $M_I(\hat p_I,\hat\theta_I)\geq M_I(p_0,\theta_0)-o_p(1)$ as $I\to\infty$. Then, the MLE is consistent.

\begin{enumerate}
    \item[(1)] Uniform consistency is proven via Corollary 2.2 of \cite{newey1991uniform}. Specifically, under the assumption that the true parameters live in the compact space $\Theta$, we must prove that $M_I$ converges pointwise to $M$ and that $\forall (p,\theta),(p',\theta')\in\Theta$, $|M_I(p',\theta')-M_I(p,\theta)|\leq O_p(1)||(p',\theta'),(p,\theta)||_1$. 
    
    We start with pointwise convergence: Let $(p,\theta)\in\Theta$ be arbitrary but fixed. Then,
    
    \begin{align*}
        M_I&(p,\theta)-M(p,\theta) \\
        &= -\theta(\overline{d}_{R,J} - E \overline{d}_{R,J}) - \log\frac{\psi(\theta)}{\psi(\theta)} + \sum_{j=1}^J \log p_j (\overline{X}_j-E\overline{X}_j) + \log(1-p_j)(E\overline{X}_j - \overline{X}_j)\\
        &= -\theta o_p(1) + \sum_{j=1}^J \log p_j o_p(1) + \log(1-p_j) o_p(1)\\
        &= o_p(1),
\end{align*}

so $M_I$ converges pointwise to $M$. Next, let $(p,\theta),(p',\theta')\in\Theta$ be arbitrary but fixed. Then,

\begin{align}
    |&M_I(p',\theta')-M_I(p,\theta)| \nonumber \\
    &= \big|\theta'\overline{D}(\Pi,\text{Order}(p'))-\theta \overline{D}(\Pi,\text{Order}(p')) +\log\frac{\psi(\theta')}{\psi(\theta)} +\sum_{j=1}^J \overline{X}_j\log \frac{p_j}{p_j'} + (M-\overline{X}_j)\log\frac{1-p_j}{1-p_j'}\big| \label{eq:consis1}\\
    &\leq \underbrace{\big|\theta'\overline{D}(\Pi,\text{Order}(p'))-\theta \overline{D}(\Pi,\text{Order}(p'))\big|}_\text{Term 1} + \underbrace{\big|\log\frac{\psi(\theta')}{\psi(\theta)}\big|}_\text{Term 2} +\underbrace{\sum_{j=1}^J \big|\overline{X}_j\log \frac{p_j}{p_j'} + (M-\overline{X}_j)\log\frac{1-p_j}{1-p_j'}\big|}_\text{Term 3} \label{eq:consis2}\\
    &= O_p(1)|\theta'-\theta| + O_p(1)|\theta'-\theta| + \sum_{j=1}^JO_p(1) |p_j'-p_j| \label{eq:consis3}\\
    &= O_p(1)||(p',\theta'),(p,\theta)||_1 \nonumber,
\end{align}
where line \ref{eq:consis1} holds by definition, line \ref{eq:consis2} holds by the triangle inequality, and line \ref{eq:consis3} is shown for each numbered term below:
\begin{align*}
&\text{Term 1} \\
    &=\big|\theta'\overline{D}(\Pi,\text{Order}(p'))-\theta \overline{D}(\Pi,\text{Order}(p'))\big|\\
    &= \big|\theta'(\overline{D}(\Pi,\text{Order}(p'))-E[\overline{D}(\Pi,\text{Order}(p'))]+E[\overline{D}(\Pi,\text{Order}(p'))] - E[\overline{D}(\Pi,\text{Order}(p_0))]) -\\
    & \ \ \ \ \ \theta(\overline{D}(\Pi,\text{Order}(p'))-E[\overline{D}(\Pi,\text{Order}(p'))] + E[\overline{D}(\Pi,\text{Order}(p'))] - E[\overline{D}(\Pi,\text{Order}(p_0))]) + \\
    & \ \ \ \ (\theta'-\theta)E[\overline{D}(\Pi,\text{Order}(p_0))]\big|\\
    &= \big|\theta'(o_p(1) + O_p(1)) -\theta(o_p(1) + O_p(1))] + (\theta'-\theta)O_p(1)\big|\\
    &= \big|(\theta'-\theta)\big|O_p(1).
\end{align*}
Next,
\begin{align*}
\text{Term 2} &= |\log\psi(\theta')-\log\psi(\theta)| \\
    &\leq C_1|\theta'-\theta|\\
    &= O_p(1)|\theta'-\theta|,
\end{align*}
since $\log\psi(\theta)$ is a continuous function defined on a compact range (specifically, the range of $\theta\in\Theta$) and therefore must have a maximum and minimum slope over that range. Similarly,
\begin{align*}
\text{Term 3} &= \sum_{j=1}^J \big|\overline{X}_j(\log p_j - \log p_j') + (M-\overline{X}_j)(\log(1-p_j)-\log(1-p_j'))\big|\\
    &\leq \sum_{j=1}^J \overline{X}_jC_2\big|p_j - p_j'\big| + (M-\overline{X}_j)C_3\big|p_j-p_j'\big|\\
    &= \sum_{j=1}^J \big|p'-p\big|O_p(1).
\end{align*}
since $\log(p_j)$ and $\log(1-p_j)$ are also continuous functions on the compact range of $p_j\in\Theta$ for each $j=1,\dots,J$.

Therefore, the estimator is uniformly consistent.

\item[(2)] The well-separation condition on $M(p_0,\theta_0)$ is straightforward to prove using Lemma 2.2 of \cite{newey1994large}: Since the model is identified (as shown in Proposition \ref{prop:ident}) and $E|\ell_{p,\theta}(X,\Pi)|<\infty$ $\forall p,\theta\in\Theta$ due to the constrained domain of each $X_{ij}$ and $\Pi_i$ regardless of $p,\theta\in\Theta$, we have the desired result.

\item[(3)] The third condition that the sequence $(\hat p_I,\hat\theta_I)$ as $I\to\infty$ satisfies $M_I(\hat p_I,\hat\theta_I) \geq M_I(p_0,\theta_0)-o_p(1)$ is a standard property of the MLE.
\end{enumerate}
Thus, according to the conditions of Theorem 5.7 in \cite{van2000asymptotic}, $(\hat p_I,\hat\theta_I)$ is consistent for $(p_0,\theta_0)$.

\section{Additional Mean-Variance Calculations in Grant Panel Review Scores} \label{sect:A_mv}

In this section, we display the relationship between sample mean and variance of real grant panel review scores to test the appropriateness of the Binomial score model. Discussion of Figure \ref{fig:mv} can be found in Section \ref{sect:realdata}.

\begin{figure}[h!!]
    \centering
    \includegraphics[width=\textwidth]{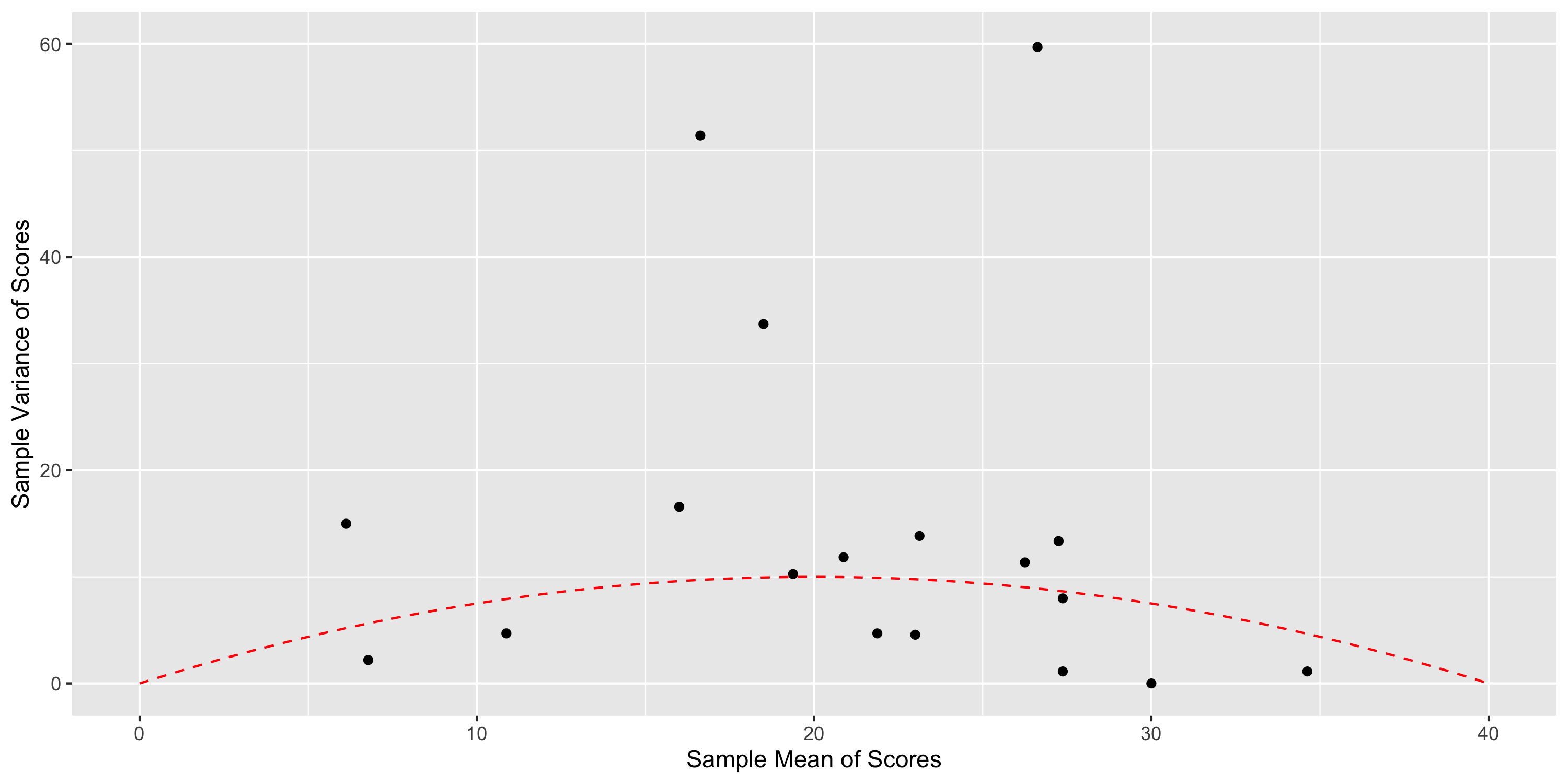}
    \caption{Sample variance of scores vs. sample mean of scores by proposal (\textit{black circles}); theoretical relationship between mean and variance in a Binomial distribution (\textit{red dotted line}).}
    \label{fig:mv}
\end{figure}
\newpage

\bibliography{refs}

\end{document}